\documentclass[a4paper,onecolumn,11pt,accepted=2025-02-08]{quantumarticle}
\pdfoutput=1
\usepackage[affil-it]{authblk}
\usepackage[usenames,dvipsnames]{xcolor}
\usepackage{amsfonts}
\usepackage{amsmath,amsthm,amssymb,dsfont}
\usepackage{enumerate}
\usepackage{graphicx}	
\usepackage{subcaption}
\usepackage[margin=3cm]{geometry}
\usepackage{url}
\usepackage{todonotes}
\usepackage{bbm}

\usepackage{tikz}

\usepackage{pifont}
\usepackage{multirow}
\usepackage{makecell}

\usepackage{epsfig}
\usetikzlibrary{shapes.symbols,patterns} 
\usepackage{pgfplots}
\pgfplotsset{compat=1.10}
\usepgfplotslibrary{fillbetween}

\definecolor{linkblue}{HTML}{001487}
\usepackage{hyperref}
\hypersetup{colorlinks=true,citecolor=linkblue,linkcolor=linkblue,filecolor=linkblue,urlcolor=linkblue,breaklinks=true}

\usepackage{nicefrac}
\usepackage{mathtools}
\usepackage[nameinlink,capitalize,noabbrev]{cleveref}

\usepackage{algorithm}
\usepackage{algorithmic}

\usepackage{mdframed}
\usepackage{aligned-overset}

\usepackage{dsfont}
\usepackage{qcircuit}

\theoremstyle{plain}
\newtheorem{theorem}{Theorem}[section]
\newtheorem{lemma}[theorem]{Lemma}

\newtheorem{corollary}[theorem]{Corollary}

\newtheorem{question}{Question}
\crefname{question}{Question}{Questions}

\theoremstyle{definition}
\newtheorem{definition}[theorem]{Definition}
\newtheorem{remark}[theorem]{Remark}
\newtheorem{example}[theorem]{Example}

\newcommand*{\ee}{\mathrm{e}}

\newcommand*{\cA}{\mathcal{A}}

\newcommand*{\cC}{\mathcal{C}}
\newcommand*{\cD}{\mathcal{D}}
\newcommand*{\cE}{\mathcal{E}}
\newcommand*{\cF}{\mathcal{F}}
\newcommand*{\cG}{\mathcal{G}}

\newcommand*{\cN}{\mathcal{N}}

\newcommand*{\cU}{\mathcal{U}}
\newcommand*{\cV}{\mathcal{V}}
\newcommand*{\cW}{\mathcal{W}}
\newcommand*{\cX}{\mathcal{X}}

\newcommand*{\N}{\mathbb{N}}

\newcommand*{\R}{\mathbb{R}}

\newcommand*{\C}{\mathbb{C}}

\newcommand*{\LO}{\mathrm{LO}}
\newcommand*{\LOCC}{\mathrm{LOCC}}

\newcommand*{\id}{\mathrm{id}}

\newcommand*{\tr}{\mathrm{tr}}
\newcommand*{\ket}[1]{| #1 \rangle}
\newcommand*{\bra}[1]{\langle #1 |}

\newcommand{\proj}[1]{|#1\rangle\!\langle #1|}

\newcommand*{\CPTP}{\mathrm{CPTP}}

\newcommand*{\CPTN}{\mathrm{CPTN}}

\newcommand*{\ag}{\mathfrak{g}}
\newcommand*{\am}{\mathfrak{m}}
\newcommand*{\ak}{\mathfrak{k}}
\newcommand*{\ah}{\mathfrak{h}}

\newcommand*{\ci}{\mathrm{i}} 

\newcommand{\norm}[1]{\left\lVert#1\right\rVert}

\newcommand*{\ox}{\otimes}

\newcommand{\Choi}{\ket{\Phi_U}}

\usepackage{arydshln}

\usepackage{float}

 \allowdisplaybreaks

\title{Cutting circuits with multiple two-qubit unitaries}

 \author{\normalsize Lukas Schmitt$^{1,2}$, Christophe Piveteau$^{1}$, and David Sutter$^{2}$}
  \affil{\small $^{1}$Institute for Theoretical Physics, ETH Zurich\\
  $^{2}$IBM Quantum, IBM Research Europe -- Zurich
 }
 \date{}

\begin{document}

\maketitle

\begin{abstract}
    Quasiprobabilistic cutting techniques allow us to partition large quantum circuits into smaller subcircuits by replacing non-local gates with probabilistic mixtures of local gates. The cost of this method is a sampling overhead that scales exponentially in the number of cuts. It is crucial to determine the minimal cost for gate cutting and to understand whether allowing for classical communication between subcircuits can improve the sampling overhead. In this work, we derive a closed formula for the optimal sampling overhead for cutting an arbitrary number of two-qubit unitaries and provide the corresponding decomposition. We find that cutting several arbitrary two-qubit unitaries together is cheaper than cutting them individually and classical communication does not give any advantage.
\end{abstract}

\section{Introduction}
To demonstrate a quantum advantage in the close future, various challenges need to be resolved such as correcting or reducing noise, dealing with the limited connectivity (on some hardware platforms), and scaling the number of available qubits.
The latter obstacle motivated the development of techniques known as \emph{circuit knitting} or \emph{circuit cutting}.
The idea is to split large circuits into subcircuits that can be executed on smaller devices. Partitioning a circuit can be achieved by performing space-like or time-like cuts, which are typically referred to as \emph{gate cuts}~\cite{Mitarai_2021,MF_21,piv23} and \emph{wire cuts}~\cite{PHOW20,BPS23,harada2023doubly,EP23}, respectively.
Circuit cutting techniques typically result in a sampling overhead that scales exponentially in the number of cuts.
One particular method of circuit knitting is based on a technique called \emph{quasiprobability simulation}.
This method works by probabilistically replacing the non-local gates across the various subcircuits by local operations.
By performing appropriate classical post-processing, the expectation value of the original large circuit can be retrieved.
More concretely, this means that we can simulate the large quantum circuit by only physically executing the small subcircuits on small quantum devices.

While quasiprobabilistic circuit cutting allows us to retrieve the expectation value of the larger circuit with arbitrary precision, the number of shots required to achieve a desired accuracy is increased compared to physically running the original circuit.
Suppose we cut a non-local unitary $U$, then the optimal achievable sampling overhead is characterized by the quantity $\gamma(U)$ which we call the $\gamma$-factor of $U$.
More precisely, the number of shots increases multiplicatively by $\gamma(U)^2$.
If we separately apply the quasiprobability simulation technique to $n$ different gates $(U_i)_{i=1}^n$ in the circuit, then the total sampling overhead behaves multiplicatively $\prod_{i=1}^n \gamma(U_i)$.
This is the reason for the exponential scaling of the sampling overhead with respect to the number of cut gates.

A central question that has been studied in previous works~\cite{piv23,Lowe2023,BPS23} is whether allowing the smaller subcircuits to exchange classical information can improve the sampling overhead.
To distinguish between the $\gamma$-factor in the two settings where classical communication is allowed or not, we introduce a corresponding subscript $\gamma_{\LO}$, respectively $\gamma_{\LOCC}$, where $\LO$ stands for \emph{local operations} and $\LOCC$ for \emph{local operations with classical communication}.
In the setting of wire cutting, it was recently proven that there is a strict separation between $\LO$ and $\LOCC$, in the sense that the $\gamma$-factor of a wire cut is strictly smaller when classical communication is allowed~\cite{BPS23}. 
However, for gate cutting it is unknown if there is a separation between the two settings.
Computing $\gamma$-factors has proven to be a difficult problem in practice, and explicit values are only known for a few gates. For these gates, there is no separation~\cite{piv23}.
Hence, it is natural to ask: 
\begin{question}\label{question1}
    Is $\LOCC$ strictly more powerful for circuit cutting than $\LO$?
\end{question}

An additional difficulty in analyzing sampling overheads in circuit cutting is that the $\gamma$-factor only characterizes the optimal overhead for cutting a single non-local gate.
When one considers cutting a circuit with multiple non-local gates, the $\gamma$-factors of the individual gates do not fully capture the optimal overhead of cutting the complete circuit at once.
For example, it has been shown~\cite{piv23} that for a Clifford gate $U$ and in the setting of $\LOCC$ it is strictly cheaper to cut $n>1$ copies of the gate that happen in the same time slice (as seen in~\cref{fig_parallelCut}), as opposed to using the optimal single-gate cutting procedure $n$ times.
Mathematically, this is captured by a strictly submultiplicative behavior of the $\gamma$-factor
\begin{equation}\label{eq:strict_submult}
    \gamma_{\LOCC}(U^{\otimes n}) < \gamma_{\LOCC}(U)^n \, ,
\end{equation}
where $U^{\otimes n}$ stands for $n$ parallel copies of the non-local gate $U$ each acting on a distinct set of qubits.
Previously to this work, it was unknown whether such a strict submultipliative behavior also exists for non-Clifford gates $U$ and/or in the absence of classical communication (i.e. with $\gamma_{\LO}$ instead of $\gamma_{\LOCC}$).
We thus ask:
\begin{question}\label{question2}
    Is cutting arbitrary gates simultaneously cheaper than cutting them separately?
\end{question}

When applying circuit cutting in practice, one will typically be confronted with the situation that the large circuit contains multiple non-local gates which occur at possibly very distant points in time.
The sub-multiplicative behavior of the $\gamma$-factor cannot be directly applied to cut multiple non-local gates that do not occur in the same time slice.
However, one can still try to apply the insight that ``global'' cutting strategies can be more efficient by dealing with multiple gates at once instead of treating them individually.

For better illustration of the different problems at hand, we introduce three separate settings depicted in~\cref{fig_cutting_scenarios}, that we will consider throughout this work.
In the \emph{single cut} setting, we consider only one single instance of some non-local gate $U$ that we want to cut.
Here, the optimal sampling overhead is fully characterized by the $\gamma$-factor of $U$.
In the \emph{parallel cut} setting, we consider multiple gates $(U_i)_{i=1}^n$ that we want to cut which all occur in the same time slice of the circuit.
Note that the parallel cut setting is a special case of the single-cut setting with $U=\bigotimes_{i=1}^n U_i$ and the optimal sampling overhead is thus fully characterized by the $\gamma$-factor of $\bigotimes_{i=1}^n U_i$.
\begin{figure}[!htb]
    \centering
    \begin{subfigure}[b]{0.2\textwidth}
        \centering
            \begin{tikzpicture}[thick,scale=0.8]
        \def \s{0.1}
     \draw [fill=gray!15,draw=none] (0.05,0.06) rectangle (2.45,1.5);   
     \draw [fill=cyan!15,draw=none] (0.05,-0.06) rectangle (2.45,-1.5);

     \draw [fill=gray!15,draw=none] (\s + 0.1,1.4+0.4) rectangle (\s+0.1+0.25,1.4+0.25+0.4);
     \draw [fill=cyan!15,draw=none] (\s + 1.6,1.4+0.4) rectangle (\s+1.6+0.25,1.4+0.25+0.4);     
     \node[gray] at (\s + 0.6,1.5+0.4) {\footnotesize{$A$}};
     \node[cyan] at (\s + 2.1,1.5+0.4) {\footnotesize{$B$}};

     \draw (0,0.5) -- (2.5,0.5);
     \draw (0,-0.5) -- (2.5,-0.5); 
     
     \draw (0,0.9) -- (2.5,0.9);
     \draw (0,-0.9) -- (2.5,-0.9); 

     \draw (0,1.3) -- (2.5,1.3);
     \draw (0,-1.3) -- (2.5,-1.3);

     \def \xx{1.25}
     \draw [fill=red!40] (\xx,0.9) circle (2.5mm);
     \node at (\xx,0.9) {\footnotesize{$U$}};
     \draw (\xx,0.9-0.25) -- (\xx,-0.9+0.25);
     \draw [fill=red!40](\xx,-0.9) circle (2.5mm);
     \node at (\xx,-0.9) {\footnotesize{$U$}};
     
     \node at (1.25,0) {$--------$};
     


    \end{tikzpicture}
                \vspace{4.5mm}
        \caption{Single cut}
        \label{fig_singleCut}
    \end{subfigure}
        \hfill
    \begin{subfigure}[b]{0.35\textwidth}
            \centering
            \begin{tikzpicture}[thick,scale=0.8]
    \def\x{0}

    \def\s{1.8}

     \draw [fill=gray!15,draw=none] (\x+0.05,0.06) rectangle (\x+6.0,1.5);   
     \draw [fill=cyan!15,draw=none] (\x+0.05,-0.06) rectangle (\x+6.0,-1.5);

     \draw [fill=gray!15,draw=none] (\s + 0.1,1.4+0.4) rectangle (\s+0.1+0.25,1.4+0.25+0.4);
     \draw [fill=cyan!15,draw=none] (\s + 1.6,1.4+0.4) rectangle (\s+1.6+0.25,1.4+0.25+0.4);     
     \node[gray] at ( \s +0.6,1.5+0.4) {\footnotesize{$A$}};
     \node[cyan] at ( \s +2.1,1.5+0.4) {\footnotesize{$B$}};



     \draw (\x,0.5) -- (\x+6,0.5);
     \draw (\x,-0.5) -- (\x+6,-0.5); 
     
     \draw (\x,0.9) -- (\x+6,0.9);
     \draw (\x,-0.9) -- (\x+6,-0.9);

     \draw (\x,1.3) -- (\x+6,1.3);
     \draw (\x,-1.3) -- (\x+6,-1.3);

     \def \xx{1.25+\x}
     \draw [fill=red!40] (\xx,1.3) circle (2.5mm);
     \node at (\xx,1.3) {\footnotesize{$U_1$}};
     \draw (\xx,1.3-0.25) -- (\xx,-0.5+0.25);
     \draw [fill=red!40](\xx,-0.5) circle (2.5mm);
     \node at (\xx,-0.5) {\footnotesize{$U_1$}};

     \def \xx{3+\x}
     \draw [fill=yellow!40] (\xx,0.9) circle (2.5mm);
     \node at (\xx,0.9) {\footnotesize{$U_2$}};
     \draw (\xx,0.9-0.25) -- (\xx,-1.3+0.25);
     \draw [fill=yellow!40](\xx,-1.3) circle (2.5mm);
     \node at (\xx,-1.3) {\footnotesize{$U_2$}};

     \def \xx{4.75+\x}
     \draw [fill=ForestGreen!40] (\xx,0.5) circle (2.5mm);
     \node at (\xx,0.5) {\footnotesize{$U_3$}};
     \draw (\xx,0.5-0.25) -- (\xx,-0.9+0.25);
     \draw [fill=ForestGreen!40](\xx,-0.9) circle (2.5mm);
     \node at (\xx,-0.9) {\footnotesize{$U_3$}};

\node at (\x+3,0) {$---------------$};



    \end{tikzpicture}
        \caption{Parallel cut}
        \label{fig_parallelCut}
    \end{subfigure}        
            \hfill
    \begin{subfigure}[b]{0.35\textwidth}
            \centering
            \begin{tikzpicture}[thick,scale=0.8]
    \def\x{0}
    
    \def \s{1.8}

     \draw [fill=gray!15,draw=none] (\x+0.05,0.06) rectangle (\x+6.0,1.5);   
     \draw [fill=cyan!15,draw=none] (\x+0.05,-0.06) rectangle (\x+6.0,-1.5);

     \draw [fill=gray!15,draw=none] (\s+0.1,1.4+0.4) rectangle (\s+0.1+0.25,1.4+0.25+0.4);
     \draw [fill=cyan!15,draw=none] (\s+1.6,1.4+0.4) rectangle (\s+1.6+0.25,1.4+0.25+0.4);     
     \node[gray] at (\s + 0.6,1.5+0.4) {\footnotesize{$A$}};
     \node[cyan] at (\s + 2.1,1.5+0.4) {\footnotesize{$B$}};

\node at (\x+3,0) {$---------------$};


     \draw (\x,0.5) -- (\x+6,0.5);
     \draw (\x,-0.5) -- (\x+6,-0.5); 
     
     \draw (\x,0.9) -- (\x+6,0.9);
     \draw (\x,-0.9) -- (\x+6,-0.9);

     \draw (\x,1.3) -- (\x+6,1.3);
     \draw (\x,-1.3) -- (\x+6,-1.3);

     \def \xx{0.8+\x}
     \draw [fill=red!40] (\xx,0.9) circle (2.5mm);
     \node at (\xx,0.9) {\footnotesize{$U_1$}};
     \draw (\xx,0.9-0.25) -- (\xx,-0.5+0.25);
     \draw [fill=red!40](\xx,-0.5) circle (2.5mm);
     \node at (\xx,-0.5) {\footnotesize{$U_1$}};

     \def \xx{2+\x}

     \draw [fill=blue!50] (\xx-0.3,1.4) -- (\xx+0.3,1.4) -- (\xx+0.5,0.4) -- (\xx+0.5,-0.4) -- (\xx+0.3,-1.4) -- (\xx-0.3,-1.4) -- (\xx-0.5,-0.4) -- (\xx-0.5,0.4) --cycle; 
     \node at (\xx,0) {\large{$\cE$}};

     \def \xx{3.1+\x}
     \draw [fill=Yellow!40] (\xx,1.3) circle (2.5mm);
     \node at (\xx,1.3) {\footnotesize{$U_2$}};
     \draw (\xx,1.3-0.25) -- (\xx,-1.3+0.25);
     \draw [fill=Yellow!40](\xx,-1.3) circle (2.5mm);
     \node at (\xx,-1.3) {\footnotesize{$U_2$}};

    \def \xx{4.3+\x}

     \draw [fill=purple!50] (\xx-0.3,1.4) -- (\xx+0.3,1.4) -- (\xx+0.5,0.4) -- (\xx+0.5,-0.4) -- (\xx+0.3,-1.4) -- (\xx-0.3,-1.4) -- (\xx-0.5,-0.4) -- (\xx-0.5,0.4) --cycle; 
     \node at (\xx,0) {\large{$\cF$}};
     
\def \xx{5.5+\x}
     \draw [fill=ForestGreen!40] (\xx,0.5) circle (2.5mm);
     \node at (\xx,0.5) {\footnotesize{$U_3$}};
     \draw (\xx,0.5-0.25) -- (\xx,-0.9+0.25);
     \draw [fill=ForestGreen!40](\xx,-0.9) circle (2.5mm);
     \node at (\xx,-0.9) {\footnotesize{$U_3$}};

    \end{tikzpicture}
        \caption{Black box cut}
        \label{fig_blackBoxCut}
    \end{subfigure}
\caption{Three different cutting scenarios, where the dashed line represents the cut. The unitaries $U_i$ are arbitrary two-qubit gates acting on $A$ and $B$. In the black box setting $\cE$ and $\cF$ denote arbitrary $\CPTP$ maps.}
\label{fig_cutting_scenarios}
\end{figure}
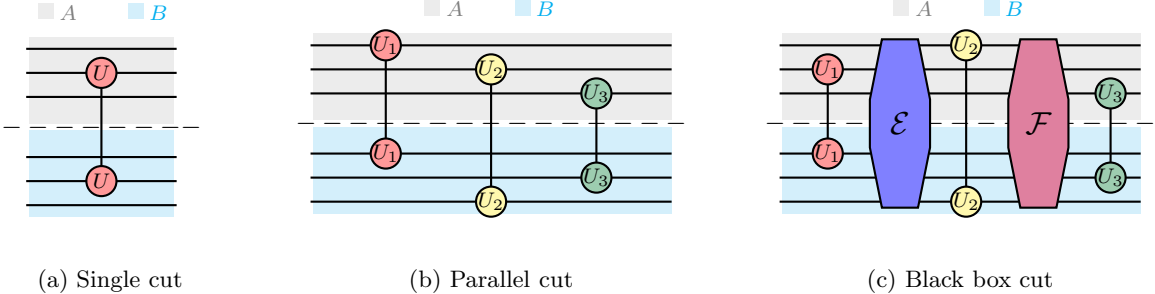

The case where we consider a circuit with multiple non-local gates $(U_i)_{i=1}^n$ that occur at different time slices is a bit more involved.
Clearly, the optimal procedure would be to consider the whole circuit (including local and non-local gates) as one big unitary $U_{\mathrm{tot}}$ and then determining the $\gamma$-factor of that.
However, this is typically intractable as just determining $U_{\mathrm{tot}}$ is as difficult as simulating the circuit itself.
For this purpose we introduce a third setting, which we call the \emph{black box} setting.
Here, we consider multiple non-local gates $(U_i)_{i=1}^n$ which we want to cut, but between these gates we may have unknown operations.
The circuit cutting technique is not allowed to have any information of what happens inside them (hence the name ``black box'') and must work for any choice of the unknown operations.

Clearly the black box setting is at least as difficult as the parallel cut setting, as the latter can be retrieved as a special choice of the black box.
It is thus natural to ask:
\begin{question}\label{question3}
    In the black box setting, can cutting the gates jointly improve the sampling overhead?
\end{question}
Here again, the answer of~\cref{question3} is known for Clifford gates under $\LOCC$~\cite{piv23}.
The gate-based teleportation scheme used by the authors makes it possible to reduce the black box setting to the parallel cut setting, showing that the optimal sampling overhead is the same in both cases.
Since gate teleportation requires classical communication and only works for Clifford gates, it was previously unknown whether the same answer to~\cref{question3} would also hold in a more general setting.

\paragraph{Results}

In this work, we answer~\cref{question1,question2,question3} for the case where the unitaries to be cut are all two-qubit gates.
It is well known, that an arbitrary two-qubit unitary $U$ exhibits a so-called \emph{KAK} decomposition
\begin{equation}\label{eq_kak_form}
    U = (V_1\otimes V_2)\left(\sum\limits_{k=0}^3 u_k \sigma_k\otimes \sigma_k\right)(V_3\otimes V_4) \, ,
\end{equation}
where $\sigma_0=\mathds{1}$, $\sigma_1,\sigma_2,\sigma_3$ are the Pauli matrices, the $u_i\in \mathbb{C}$ fulfill $\sum_{i=0}^3 \lvert u_i\rvert ^2 = 1$ and the $V_i$ are single-qubit unitaries.
For any such unitary $U$ we can define the quantity
\begin{equation}
    \Delta_U := \sum\limits_{k\neq k'} \lvert u_k\rvert \lvert u_{k'}\rvert \, ,
\end{equation}
which is nonnegative and zero only for a separable unitary.
We show in~\cref{cor_two_qubit_gamma} that the $\gamma$-factor of $U$ is given by
\begin{align} \label{eq_overhead_setting1}
    \gamma_{\LO}(U) = \gamma_{\LOCC}(U)  = 1 + 2\Delta_U .
\end{align}
Furthermore, in~\cref{cor_parallel_cuts} we characterize the $\gamma$-factor in the parallel cut setting and show that
\begin{align} \label{eq_overhead_setting2}
    \gamma_{\LO}(U^{\otimes n}) = \gamma_{\LOCC}(U^{\otimes n}) = 2(1+\Delta_U)^n - 1 < \gamma_{\LO}(U)^n = \gamma_{\LOCC}(U)^n \, .
\end{align}
As a consequence, we show in~\cref{cor_regularized_gamma} that the regularized $\gamma$-factor, i.e.~the effective sampling overhead per gate in the limit of many gates, is given by
\begin{align}
    \lim\limits_{n\rightarrow\infty}\left( \gamma_{\LO}(U^{\otimes n}) \right)^{1/n} = \lim\limits_{n\rightarrow\infty}\left( \gamma_{\LOCC}(U^{\otimes n}) \right)^{1/n} = 1 + \Delta_U \, .
\end{align}
We therefore have a full understanding of the optimal sampling overhead of two-qubit gates in the single-cut and parallel cut settings, and classical communication does not reduce the overhead in either of them.
Our results are derived using a more general characterization of the $\gamma$-factor of general (possibly larger than two-qubit) unitaries that exhibit a form analogous to~\cref{eq_kak_form}, as seen in~\cref{Master}.
Generally, our techniques can be seen as a generalization and improvement of the ideas introduced in~\cite{MF_21}.

To tackle the black box setting, we devise an explicit protocol to perform black box cutting by reducing it to the parallel setting.
Therefore, the overhead is the same as in the parallel setting (i.e. as in~\cref{eq_overhead_setting2}).
Since the black box setting is at least as difficult than the parallel setting, our result is therefore optimal.
As in the approach of~\cite{piv23}, this protocol requires additional ancilla qubits to carry information across the black boxes.
However, our protocol does not require classical communication and it also works for non-Clifford gates.

In summary, we answer the questions as follows:
\begin{description}
    \item[Answer to~\cref{question1}:] For quasiprobabilistic circuit cutting, classical communication does not improve the overhead of cutting arbitrary two-qubit gates. 
    \item[Answer to~\cref{question2}:] The answer is positive for all (non-local) two-qubit gates. In other words, the $\gamma$-factor behaves strictly submultiplicatively.
    \item[Answer to~\cref{question3}:] For arbitrary two-qubit gates, the black box setting has the same sampling overhead as the parallel cut setting (without requiring classical communication).
    This comes at the cost of four additional ancilla qubits per gate.
\end{description}

\begin{table}[!htb]
\centering
\def\arraystretch{1.4}
  \begin{tabular}{cccc}
Setting & Reference & Example (\cref{fig_cutting_scenarios})  \\
\hline
Single cut &  \Cref{cor_two_qubit_gamma} & $\gamma(\mathrm{CNOT})=3$, $\gamma(\mathrm{CR}_X(\theta)) =  1+2 |\sin(\theta/2)|$    \\
Parallel cut & \Cref{cor_parallel_cuts} & $\gamma( \mathrm{CNOT} \otimes \mathrm{CR}_X(\theta)) =  3+4|\sin(\theta/2)|$   \\
Black box cut & \Cref{theorem:blackbox} & $\bar \gamma( \mathrm{CNOT},\mathrm{CR}_X(\theta)) =  3+4|\sin(\theta/2)|$ 
  \end{tabular}
		\caption{Summary on how to optimally cut two-qubit gates. The example considers the setting from~\cref{fig_cutting_scenarios} with $U_1=\mathrm{CNOT}$, $U_2=\mathrm{CR}_X(\theta)$, and $U_3=\mathds{1}$. The KAK parameters of these unitaries are given in~\cref{ex_2qubit}. Note that for these examples, we write $\gamma$ without subscript, because the overhead for $\LO$ and $\LOCC$ is the same. Our method in the black box setting requires additional ancillary qubits, however this is not the case for the other two settings.}
		\label{tab_results}
\end{table}
The results together with a small example are summarized in~\cref{tab_results}.
The example given in~\cref{tab_results} shows that the submultiplicative property of the $\gamma$-factor does not just appear when cutting multiple identical gates, but also when one cuts multiple distinct gates:
\begin{align} \label{eq_example}
\gamma\big( \mathrm{CNOT} \otimes \mathrm{CR}_X(\theta)\big) = 3+4|\sin(\theta/2)| < 3+6 |\sin(\theta/2)|\big) = \gamma(\mathrm{CNOT})\gamma\big(\mathrm{CR}_X(\theta)\big) \, .
\end{align}
The difference between the two sides of~\cref{eq_example} can be substantial.
The submultiplicativity of the $\gamma$-factor becomes particularly meaningful if we cut more than just two gates simultaneously.
\Cref{fig_haar} shows the sampling overhead for a two-qubit unitary that is sampled from the Haar measure. 
\begin{figure}[H]
    \centering
    \begin{subfigure}[b]{0.45\textwidth}
        \centering
          \begin{tikzpicture}[scale=0.85]
	\begin{axis}[
		height=5.0cm,
		width=8.0cm,
		grid=major,
		xlabel=$\gamma$-factor,
		ylabel=freqency,
		xmin=2,
		xmax=7,
		ymax=0.6,
		ymin=0,
	     ytick={0,0.2,0.4,0.6},
          xtick={2,3,4,5,6,7},
	]

    \addplot[dotted,smooth,name path=f] table {hist.dat};

 	\addplot[dashed,smooth,name path=f] coordinates {
    (5.71,0) (5.71,7)
  };
	
	\end{axis} 
\end{tikzpicture}
        \caption{Distribution of $\gamma$-factor for a Haar-randomly sampled two-qubit gate $U$: $\gamma(U)$ is plotted on the x-axis. The y-axis shows the relative frequency out of $10^7$ samples. The dashed line shows the average $\gamma$-factor which is approximately $\gamma(U) \approx 5.71$. }
        \label{fig_haar_distribution}
    \end{subfigure}
        \hfill
    \begin{subfigure}[b]{0.45\textwidth}
            \centering
              \begin{tikzpicture}[scale=0.85]
	\begin{axis}[
		height=5.0cm,
		width=8.0cm,
		grid=major,
        ymode=log,
		xlabel=number of gates,
		ylabel=overhead,
		xmin=1,
		xmax=10,
		ymax=100000000,
		ymin=5,
	     ytick={10,100,1000,10000,100000,1000000,10000000},
          xtick={1,2,3,4,5,6,7,8,9,10},
		legend style={at={(0.3,1)},anchor=north,legend cell align=left,font=\footnotesize} 
	]

	\addplot[dashed,smooth,name path=f] coordinates {
(1.000000000000000000, 5.709807069222020845)
(2.000000000000000000, 32.6018967677377631) 
(3.000000000000000000, 186.150540634475727) 
(4.000000000000000000, 1062.88367285423109) 
(5.000000000000000000, 6068.86070902375135) 
(6.000000000000000000, 34652.0237785075988) 
(7.000000000000000000, 197856.370333372149) 
(8.000000000000000000, 1129721.70202009784) 
(9.000000000000000000, 6450492.96044789789) 
(10, 36831070.3055322698)
	};
	\addlegendentry{single cut}
	
		\addplot[dotted,thick,smooth,name path=h] coordinates {
(1.000000000000000000, 5.709807069222020845)
(2.000000000000000000, 21.51075545309090077)
(3.000000000000000000, 74.5214130363387639)
(4.000000000000000000, 252.367055534430874)
(5.000000000000000000, 849.022030166446689)
(6.000000000000000000, 2850.74191350263951)
(7.000000000000000000, 9566.31902540836680)
(8.000000000000000000, 32096.4324150937291)
(9.000000000000000000, 107682.789461335829)
(10.000000000000000000, 361267.725884143705)
	};
		\addlegendentry{joint cut}
	
	\end{axis} 
\end{tikzpicture}
        \caption{Overhead for single and parallel cut for multiple Haar-random gates.
        The x-axis shows the number of gates to cut. The y-axis shows the overhead. Both methods scale exponentially in the number of gates, but the base of the parallel cut is lower.}
        \label{fig_scaling}
    \end{subfigure}
\caption{Submultiplicativity of $\gamma$-factor for Haar-random two-qubit unitaries.}
\label{fig_haar}
\end{figure}
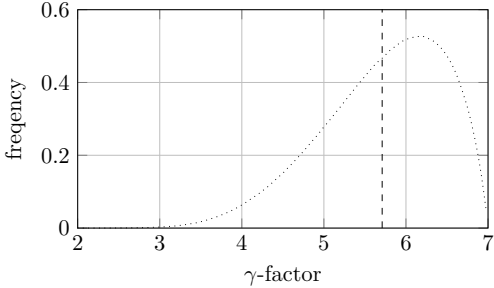
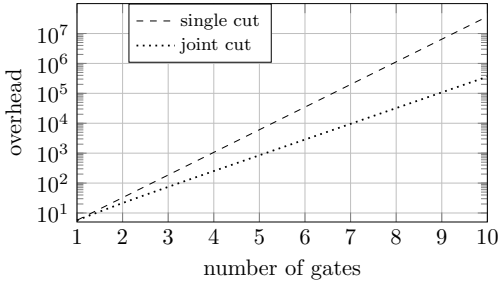


\section{Preliminaries}\label{sec_preliminaries}
\subsection{Quasiprobability simulation}
Here we briefly sketch the quasiprobability simulation framework on which our circuit cutting method is based.
For more thorough explanations, we refer the reader to~\cite{piv23,TBG17,endo18,Mitarai_2021,Piv_masterThesis}.

Consider that we are given a circuit that prepares some state and measures an observable, and we want to compute the expectation value of that observable.
Furthermore, that circuit includes some gate $\cE$ that we cannot physically run on our computer.
In the context of gate cutting, we can only physically run operations that act locally on the two circuit bipartitions, and $\cE$ would correspond to the channel of some non-local gate.

Quasiprobability simulation is a method of simulating a linear operation $\cE$ via a set of available operations $S$ that can be physically realized. 
Concretely, we can simulate an operation $\cE$ if we can decompose it as
\begin{equation}\label{eq_qpd_decomp}
    \cE = \sum_{i=1}^m a_i \mathcal{F}_i\, ,
\end{equation}
where $a_i \in \R$ and $\mathcal{F}_i \in S$. 
This decomposition can be rewritten as
\begin{equation}\label{eq_qpd_decomp_rewritten}
    \cE = \sum_{i=1}^m p_i \mathcal{F}_i \cdot \mathrm{sign}(a_i)(\sum_i |a_i|) \, ,
\end{equation}
for the probability distribution $p_i = |a_i|/(\sum_i |a_i|)$.
\Cref{eq_qpd_decomp_rewritten} implies that the expectation value of the circuit can be obtained with following Monte Carlo sampling technique:
Every shot of the circuit, a random index $i$ is sampled according to the probability distribution $(p_i)_i$.
Then, the gate $\cE$ is replaced by the operations $\cF_i$ and this modified circuit is physically executed on the hardware.
The final measurement outcome is multiplied by $\mathrm{sign}(a_i)(\sum_i |a_i|)$ and stored in memory.
By repeating this process many times and averaging the resulting outcomes,~\cref{eq_qpd_decomp_rewritten} implies that one will obtain the expectation value of the original circuit.

However, because the post-processing boosts the magnitude of the measurement outcomes by $(\sum_i |a_i|)$, the method incurs an additional sampling overhead if one wants to estimate the expectation value to the same accuracy.
More concretely, to estimate the expectation value to an additive error of at most $\epsilon$ with a probability of at least, $1-\delta$, Hoeffding's bound implies that this can be achieved with a number of shots equal to $C\cdot (\sum_i |a_i|)^2 \epsilon ^{-2}\ln \left( \frac{1}{2\delta}\right)$ for some constant $C$.
So clearly, if one wants to use the quasiprobability simulation technique, one should use a decomposition that exhibits the smallest possible value of $(\sum_i |a_i|)$.
Since the simulation overhead scales with $(\sum_i |a_i|)^2$, we assign a special name for the smallest achievable value of $(\sum_i |a_i|)$:
\begin{definition} \label{def_gamma_factor}
  The $\gamma$-factor of an operation $\cE$ over $S$ is defined as
  \begin{align}\label{eq_def_gamma}
      \gamma_S(\mathcal{E}) \coloneqq \min \Big\{ \sum_{i=1}^m \lvert a_i \rvert :  \mathcal{E} = \sum\limits_{i=1}^m a_i \mathcal{F}_i \text{ where } m\geq 1, \mathcal{F}_i\in S \textnormal{ and } a_i \in \R \Big\} \, .
  \end{align}
\end{definition}

A particular application of quasiprobability simulation is gate cutting~\cite{Mitarai_2021,piv23}.
Consider a bipartition of the qubits in the circuit into two systems $A$ and $B$ corresponding to the subcircuits obtained from the circuit cutting.
The operation $\cE$ is then a unitary channel $\cE(\rho):=U\rho U^{\dagger}$ for some unitary $U$ that acts non-locally across $A$ and $B$.
As a slight abuse of notation, we sometimes denote the $\gamma$-factor of a unitary channel simply as $\gamma_S(U)$.
The set of operations $S$ into which we decompose our gate is then chosen to be either local operations on the systems $A$ and $B$ (denoted by $\LO(A,B)$) or local operations with classical communication ($\LOCC(A,B)$).
If this quasiprobability simulation is applied for all gates acting non-locally across $A$ and $B$, we can thus replace the large non-local circuit with a probabilistic mixture of local circuits on $A$ and $B$ that do not have any entangling gates between them (plus some additional classical post-processing).
The method therefore allows us to obtain the expectation value of the original circuit by only physically executing smaller circuits.

Following~\cite{piv23}, we defined the set of local operations $\LO(A,B)$ across the systems $A,B$ as follows:
\begin{definition}
    The set $\LO(A,B)$ is given by $\{\cF_A \ox \cF_B \mid \cF_A \in \cN(A) \, \text{and} \, \cF_B \in \cN(B)\}$, where $\cN(X)$ is defined as\footnote{It is even sufficient to only demand $\cA^+ + \cA^- \in \CPTN(X)$, but this can be achieved in the $\CPTP$ setting by postweighting with $0$. }
    \begin{equation}
        \cN(X) = \left\{\cA^+ - \cA^- \mid \cA^+,\cA^- \text{completely positive},  \cA^+ + \cA^- \in \CPTP(X) \right\}\, .
    \end{equation}
    where $\CPTP(X)$ is the set of all completely-postive trace-preserving maps on the system $X$.
\end{definition}
This is a slightly non-standard definition of non-local operations, since it allows for non-completely positive maps.
The reason for this is that this definition captures the idea, first seen in \cite{Mitarai_2021} and later expanded upon in \cite{zhao2023power}, that the individual operations $\mathcal{F}_i$ can be realized as quantum instruments, and depending on which component of the instrument is executed, the measurement outcome at the end of the circuit is weighted with $+1$ or $-1$. 
As an illustrative example, consider two completely-positive trace-nonincreasing maps $\cA^+$ and $\cA^-$ s.t. $\cA^++\cA^- \in \CPTP$.
One can simulate the map $\cA^+ - \cA^-$ via the $\CPTP$-map $\cF$
\begin{equation}
    \cF(\rho) = \cA^+(\rho) \ox \ket{0}\bra{0} + \cA^-(\rho) \ox \ket{1}\bra{1}
\end{equation}
followed by a measurement of the ancillary qubit in the Z-basis.
Depending on the ancilla measurement outcome, one then weights the final measurement outcome of the circuit by $\pm1$. We refer to this as the \emph{negativity} trick.

A characterization of the set of $\LOCC(A,B)$ is very complicated, since it can contain multiple rounds of operations and communications between A and B, so we refrain from writing the full definition here.
For an in-depth definition we refer the reader to~\cite{CLMOW14}.
Since $\LOCC$ is usually defined as a quantum instrument, we can realize the same post-processing trick as with $\LO$ above and allow for non-completely positive maps by weighting some of the element of the instrument by $-1$.


\subsection{Representation of two-qubit unitaries}\label{ssec:kak}
It is well known that a semi-simple Lie algebra $\ag$ can be decomposed as $\ag = \am \oplus \ak$, $\am = \ak^\perp$, with the following relations
\begin{align}
    [\ak,\ak] \subset \ak\, , \qquad 
    [\ak,\am] = \am \, , \qquad \textnormal{and} \qquad 
    [\am,\am] \subset \ak \, .
\end{align}
This decomposition is called Cartan decomposition. For such a pair, one can find the corresponding Cartan subalgebra $\ah$, which is the largest subalgebra that is contained in $\am$. 
For these algebras, it is well known, that the Lie group $G$ can be written as\begin{equation}
    G = KAK \, ,
\end{equation}
where $K=\exp(\ak)$ and $A=\exp(\ah)$ (see for instance \cite{KHANEJA200111}). This decomposition is also called the KAK decomposition. 

Applied to $\ag = \mathfrak{su}(2^n)$, we find $\ak = \mathfrak{su}(2^{n-1}) \oplus \mathfrak{su}(2^{n-1})$, which is in general not useful for circuit cutting. However, in the case of $\ag = \mathfrak{su}(4)$, this simplifies to $\ak = \mathfrak{su}(2) \oplus \mathfrak{su}(2)$ and therefore $K = SU(2)\ox SU(2)$. These elements can be implemented with local operations. For the cartan subalgebra one finds
$\ah = \{\ci \sigma_1 \ox \sigma_1, \ci \sigma_2 \ox \sigma_2,\ci \sigma_3 \ox \sigma_3\}$, where we denoted the Pauli matrices by $\sigma_1 = X$, $\sigma_2 = Y$, $\sigma_3 = Z$. Furthermore let $\sigma_0=\mathds{1}$ denote the identity matrix. 

Using the above, it follows that any two-qubit unitary U(4) has a KAK decomposition
\begin{align} \label{eq_standard_dec}
    U = \left(V_1\otimes V_2\right)  \exp\left( \sum_{k=1}^3 \ci\theta_k \sigma_k \otimes \sigma_k \right) \left(V_3\otimes V_4\right)  \, ,
\end{align}
for some single-qubit unitaries $V_1,V_2,V_3,V_4$, and $\theta_k\in\mathbb{R}$ such that $|\theta_3| \leq \theta_2 \leq \theta_1 \leq \pi/4$.
Evaluating the exponential function, this can be rewritten as
\begin{align}  \label{eq_standard_dec_u}
    U 
    =\left(V_1\otimes V_2 \right)  \left( \sum_{k=0}^3 u_k \sigma_k \otimes \sigma_k\right)  \left(V_3\otimes V_4\right)
    =:\left(V_1\otimes V_2 \right)  W \left(V_3\otimes V_4\right) \, ,
\end{align}
where $u_k \in\mathbb{C}$ depend on $\theta_1,\theta_2,\theta_3$ and fulfill $\sum_{i=0}^3 |u_i|^2 = 1$.
\Cref{eq_standard_dec_u} is central in our analysis of optimal gate cutting of two-qubit unitaries, as cutting $U$ is equivalent to cutting the interaction unitary $W$.
In particular we have $\gamma_{S}(U) = \gamma_S(W)$ and in the proof of~\cref{Master}, we thus focus on cutting $W$.
Since there does not exist an analogous decomposition to~\cref{eq_standard_dec_u} for general unitaries, optimal cutting still remains an open problem.


\section{Cutting one single two-qubit unitary} \label{sec_single_2qubit}
In this section we show how to optimally cut an arbitrary two-qubit gate.
\begin{corollary} \label{cor_two_qubit_gamma}
Let $U$ be a two-qubit unitary with KAK parameters $(u_0,u_1,u_2,u_3) \in \C^4$. Then, 
\begin{align}
\gamma_{\LO}(U)=\gamma_{\LOCC}(U) = 1 + 2 \sum_{i\ne j} |u_i| |u_j| \, . 
\end{align}
\end{corollary}
\begin{proof}
    \cref{cor_two_qubit_gamma} follows as a special case of~\cref{Master}. Since every two-qubit unitary can be written as in~\cref{eq_standard_dec_u}, we can directly apply~\cref{Master} with $L_k=\sigma_k$ and $R_k=\sigma_k$. The coefficients $u_k$ are directly given through the KAK coefficients of $U$. The optimal decomposition of $U$ can be found in the proof of~\cref{Master}. 
\end{proof}

The result proves that there is no advantage for cutting a two-qubit unitary with $\LOCC$ compared to the $\LO$ setting. In the following, whenever this is the case, we will drop the subscript of the $\gamma$-factor.
The previously best known decomposition~\cite{Mitarai_2021} had a sampling overhead of $1+\sum_{i\ne j} (|u_i u_j^* + u_j u_i^*| + |u_i u_j^* - u_j u_i^*|)$, which is only optimal for certain gates.

\begin{example} \label{ex_2qubit}
Let us apply~\cref{cor_two_qubit_gamma} to some well-known two-qubit unitaries to check consistency with~\cite{piv23}:
\begin{enumerate}[(i)]
\item $\gamma(\mathrm{CNOT})=3$, since $u=\left(\frac{1}{\sqrt{2}},\frac{\ci}{\sqrt{2}},0,0\right)$.
\item $\gamma(\mathrm{SWAP})=7$, since $u=\left(\frac{\sqrt{2}(1+\ci)}{4},\frac{\sqrt{2}(1+\ci)}{4},\frac{\sqrt{2}(1+\ci)}{4},\frac{\sqrt{2}(1+\ci)}{4}\right)$.
\item $\gamma(\mathrm{CR}_X(\theta)) =  1+2 |\sin(\theta/2)|$, since  $u_{\mathrm{CR}_X(\theta)}=(\cos(\frac{\theta}{4}),\ci \sin(\frac{\theta}{4}),0,0)$.\footnote{$\mathrm{CR}_X(\theta)$ denotes a controlled $\mathrm{R}_X(\theta)$ gate with rotation angle $\theta$, where $\mathrm{R}_X(\theta):=\ee^{-\ci \frac{\theta}{2} X}$.}
\end{enumerate}
\end{example}

\begin{corollary}\label{co_bounded}
Let $U$ be a two-qubit unitary, then $1 \leq \gamma(U)\leq 7$.
\end{corollary}
\begin{proof}
Clearly $\gamma(U)$ cannot be smaller than $1$.
To see that it is bounded from above by $7$, consider KAK parameters $u:=(u_0,u_1,u_2,u_3) \in \C^{4}$. As a consequence of unitarity, these fulfill $\norm{u}_2=1$.
\Cref{cor_two_qubit_gamma} allows us to write 
\begin{align}
    \gamma(U)
    &= 1 + 2 \sum_{i \ne j} |u_i| |u_j| \\
    &= \sum_{i=0}^3 |u_i|^2 + 2 \sum_{i\neq j} |u_i| |u_j| \\ 
    &= 2\left(\sum_{i=0}^3 |u_i|\right)^2  - \sum_i |u_i|^2 \\
    &= 2 \norm{u}_1^2  - 1  \\
    &\leq 7  \, ,
\end{align}
where the final step uses the Cauchy-Schwarz inequality, i.e., $\norm{u}_1 \leq \sqrt{4} \norm{u}_2$.
\end{proof}

\section{Cutting parallel two-qubit unitaries}\label{sec_parallel_2qubit}
In~\cref{sec_single_2qubit} we have seen how to optimaly cut an arbitrary two-qubit unitary. In this section we would like to extend the setting to cutting multiple two-qubit unitaries which are arranged in parallel (cf.~\cref{fig_parallelCut}).

\begin{corollary} \label{cor_parallel_cuts}
 Let $n \in \N$ and $(U^{(i)}_{A_i B_i})_{i=1}^n$ be a family of two-qubit unitaries with KAK parameters $(u^{(i)}_0,u^{(i)}_1,u^{(i)}_2,u^{(i)}_3)_{i=1}^n$. 
 Then, for $U_{AB} = \otimes_{i=1}^n U^{(i)}_{A_i B_i}$, where $A=\otimes_{i=1}^n A_i$ and $B=\otimes_{i=1}^n B_i$ we have 
    \begin{equation}
        \gamma_{\LO(A:B)}(U_{AB}) = \gamma_{\LOCC(A:B)}(U_{AB}) = 1 + 2 \sum_{\mathbf{k} \neq\mathbf{k}'} |u_{\mathbf{k}}| |u_{\mathbf{k}'}| \, ,
    \end{equation}
    where $\mathbf{k} \in \{0,1,2,3\}^n$ and $u_{\mathbf{k}} = \prod_{i=1}^n  u^{(i)}_{k_i}$.
\end{corollary}
\begin{proof}
    \cref{cor_parallel_cuts} follows as a consequence of~\cref{Master}. The KAK decomposition from~\cref{eq_standard_dec_u} allows us to write for any $i \in \{1,\dots,n\}$
    \begin{equation}
        U_{A_i B_i}^{(i)} = \left(V_{A_i}^{(1,i)}\ox V_{B_i}^{(2,i)} \right)  \left( \sum_{k=0}^3 u_k^{(i)} (\sigma_k)_{A_i} \ox (\sigma_k)_{B_i} \right)  \left(V_{A_i}^{(3,i)} \ox V_{B_i}^{(4,i)} \right) \, .
    \end{equation}
This implies \small
    \begin{align}
&U_{AB} \nonumber \\
&=\bigotimes_{i=1}^n U^{(i)}_{A_i B_i} \\
      &= \Big(\! \bigotimes_{i=1}^n V_{A_i}^{(1,i)}\!\ox\! \bigotimes_{i=1}^n V_{B_i}^{(2,i)}\!\Big) \! \sum_{k_1,\ldots,k_n=0}^3 \left[\! \Big(\!\prod_{j=1}^n u^{(j)}_{k_j}\!\Big) \!\Big(\bigotimes_{i=1}^n \sigma_{k_i} \Big)_A \!\otimes \!\Big(\!\bigotimes_{i=1}^n \sigma_{k_i}\! \Big)_B \!\right] \Big(\! \bigotimes_{i=1}^n V_{A_i}^{(3,i)}\!\ox\! \bigotimes_{i=1}^n V_{B_i}^{(4,i)}\!\Big) \, ,
    \end{align}
    \normalsize
    and fulfills therefore the assumptions of~\cref{Master}. 
\end{proof}

The expression of the $\gamma$-factor simplifies further in the case that many of the unitaries $U^{(i)}$ are identical.
For instance, when there are three different types of unitaries, the $\gamma$-factor behaves as follows (the same result holds analogously for a number of gates different than three):

\begin{corollary}[Calculating the overhead]\label{cor_equation}
    Let $U,V,W$ be two-qubit unitaries and $n,m,p \in \N$. Then,
    \begin{equation} 
        \gamma(U^{\ox n}\ox V^{\ox m}\ox W^{\ox p}) 
    = 2  \Big(\sum_{i=0}^3 |u_i|\Big)^{2n} \Big(\sum_{i=0}^3 |v_i|\Big)^{2m} \Big(\sum_{i=0}^3 |w_i|\Big)^{2p} - 1 \, ,
    \end{equation}
    where $u_i,v_i$ and $w_i$ are the corresponding KAK parameters of $U$, $V$, and $W$, respectively.
\end{corollary}
\begin{proof}
    As seen in \cref{cor_parallel_cuts}, we can calculate $\gamma(T)$ as
    \begin{equation}
        \gamma(T) 
        = 1 + 2 \sum_{\mathbf{k}\ne\mathbf{k'} } |t_\mathbf{k}||t_\mathbf{k'}|
        =2\left(\sum_\mathbf{k} |t_\mathbf{k}| \right)^2 -1 \, .
    \end{equation}
    If $T=U^{\ox n}\ox V^{\ox m}\ox W^{\ox p}$, then $|t_\mathbf{k}|$ is given by
    \begin{equation}
        |t_\mathbf{k}| = \Big(\prod_{r=1}^n |u_{i_r}| \Big) \Big( \prod_{r=1}^m  |v_{j_r}| \Big) \Big( \prod_{r=1}^p |w_{l_r}|\Big) \, .
    \end{equation}
    We then calculate
    \begin{align}
         \quad \sum_\mathbf{k} |t_\mathbf{k}|      
         &= \sum_{i_1,\dots,i_n,j_1,\dots,j_m,l_1,\dots,l_p = 0}^3 |u_{i_1}|\cdots|u_{i_n}|\cdot |v_{j_1}|\cdots|v_{j_m}| \cdot |w_{l_1}|\cdots|w_{l_p}| \\
        &= \! \left(\sum_{i_1=0}^3 |u_{i_1}|\! \right)\!\cdots\!\left(\sum_{i_n=0}^3 |u_{i_n}|\! \right)\! \left(\sum_{j_1=0}^3 |v_{j_1}| \!\right)\!\cdots\!\left(\sum_{j_m=0}^3 |v_{j_m}| \!\right)\!\left(\sum_{l_1=0}^3 |w_{l_1}|\! \right)\!\cdots\!\left(\sum_{l_p=0}^3 |w_{l_p}| \!\right) \\
        &= \left(\sum_{i=0}^3 |u_{i}| \right)^n\left(\sum_{j=0}^3 |v_{j}| \right)^m\left(\sum_{l=0}^3 |w_{l}| \right)^p \, ,
    \end{align}
    which concludes the proof.
\end{proof}

\begin{example} \label{ex_parallel_cuts}
Let us apply~\cref{cor_equation} to some interesting examples:
\begin{enumerate}[(i)]
\item $\gamma(\mathrm{CNOT}^{\otimes n}) = 2^{n+1}-1$, since $u_{\mathrm{CNOT}}=(\frac{1}{\sqrt{2}},\frac{\ci}{\sqrt{2}},0,0)$ (see~\cref{ex_2qubit}).  
    This is consistent with the LOCC result from~\cite{piv23}.
\item $\gamma(\mathrm{SWAP}^{\otimes n}) = 2 \times 4^{n}-1$, since $u_{\mathrm{SWAP}}=(\frac{\sqrt{2}(1+\ci)}{4},\frac{\sqrt{2}(1+\ci)}{4},\frac{\sqrt{2}(1+\ci)}{4},\frac{\sqrt{2}(1+\ci)}{4})$ (see~\cref{ex_2qubit}).  
    This is consistent with the LOCC result from~\cite{piv23}.
\item $\gamma(\mathrm{CR}_X(\theta)^{\otimes n}) =  2 (|\cos(\frac{\theta}{4})|+|\sin(\frac{\theta}{4})|)^{2n}-1$, since  $u_{\mathrm{CR}_X(\theta)}=(\cos(\frac{\theta}{4}),\ci \sin(\frac{\theta}{4}),0,0)$.\footnote{$\mathrm{CR}_X(\theta)$ denotes a controlled $\mathrm{R}_X(\theta)$ gate with rotation angle $\theta$, where $\mathrm{R}_X(\theta):=\ee^{-\ci \frac{\theta}{2} X}$.}
This answers an open question from~\cite{piv23}.
\end{enumerate}
\end{example}
\cref{ex_2qubit} and \cref{ex_parallel_cuts} show that $n$ parallel cuts can be considerably cheaper than cutting $n$-times a gate separately. 

\begin{corollary}
    In general, \cref{cor_equation} implies that for any two-qubit gate with $\gamma(U)>1$ the optimal sampling overhead is strictly submultiplicative under the tensor product, i.e. 
    \begin{equation}
        \gamma(U^{\otimes n}) < \gamma(U)^n \, .
    \end{equation}
\end{corollary}
\begin{proof}
    Introducing $x=(\sum_i |u_i|)^2$, \cref{cor_equation} is equivalent to
    \begin{equation}
        \gamma(U^{\ox n}) = 2x^n-1 < (2x-1)^n = \gamma(U)^n \, .
    \end{equation}
    Since $\gamma(U)>1$, we have $x>1$.
    The inequality then follows by induction
    \begin{align}
        (2x-1)^{n+1} &=  (2x-1)^{n} (2x-1) \\
        &> (2x^n-1) (2x-1) \\
        &= 4x^{n+1} - 2x - 2x^n +1\\
        &= 2x^{n+1} -1 + 2x^{n+1} - 2x - 2x^n + 2 \\
        &= \gamma(U^{\ox (n+1)}) + 2x^n(x-1) - 2(x-1) \\
        &= \gamma(U^{\ox (n+1)}) + 2(x^n-1)(x-1) \\
        &> \gamma(U^{\ox (n+1)}) \, .
    \end{align}
\end{proof}

In the following, we want to consider the asymptotic difference between the single-cut and parallel cut settings.
We therefore introduce the following quantity:
\begin{definition}[Regularized $\gamma$-factor]
    The regularized $\gamma$-factor is defined as
    \begin{equation}
        \gamma_R(U) = \lim_{{n \to \infty}} (\gamma(U^{\ox n}))^{\frac{1}{n}} \, .
    \end{equation}
\end{definition}
$\gamma_R(U)$ can be considered as the ``effective $\gamma$-factor'' of a single gate in the parallel cut setting, in the limit of infinite gates.

\begin{corollary}\label{cor_regularized_gamma}
    Using~\cref{cor_equation}, $\gamma_R(U)$ evaluates to
    \begin{align}
        \gamma_R(U) &= \lim_{{n \to \infty}} \left(2(\sum_i |u_i|)^{2n} -1 \right)^{\frac{1}{n}} 
        = \left(\sum_i |u_i|\right)^{2} 
        = 1 + \sum_{i\ne j} |u_i||u_j| \, .
    \end{align}
\end{corollary}


\section{Cutting KAK-like unitaries}\label{sec_master_theorem}
Previous chapters relied on the KAK-decomposition of two-qubit unitaries. As discussed in~\cref{ssec:kak}, larger unitaries can in general not be represented in such a form. However, one can still consider the class of unitaries that are in a similar form to the two-qubit KAK decomposition. For such unitaries, we formulate the following theorem:
\begin{theorem}[Cutting KAK-like unitaries]\label{Master}
    Let $U_{AB}$ be a unitary of the form 
\begin{equation}
    U_{AB} = \left(V_A^{(1)}\otimes V_B^{(2)}\right)  \sum_{k} u_{k} (L_k)_A \otimes (R_k)_B  \left(V_A^{(3)}\otimes V_B^{(4)}\right)  \, ,
\end{equation}
where $V_A^{(1)},  V_B^{(2)}, V_A^{(3)}, V_B^{(4)}, (L_k)_A$ and $(R_k)_B$ are unitaries on $A$ and $B$. 

Then it holds that 
\begin{align}
    \gamma_{\LO}(U) \leq \sum_k |u_k|^2 + 2 \sum_{k \ne k'} |u_k| |u_{k'}| \, . \label{eq_ineq}
\end{align}
Furthermore, if $(L_k)_A$ and $(R_k)_B$ are orthogonal~\footnote{Orthogonal with respect to the Hilbert-Schmidt inner product $\langle A,B\rangle =  \tr(A^\dagger B)$} for different $k$, then~\cref{eq_ineq} simplifies, holds with equality 
and we find $\gamma_{\LO}(U) = \gamma_{\LOCC}(U)$:
\begin{align}
    \gamma_{\LO}(U) = 1 + 2 \sum_{k \ne k'} |u_k| |u_{k'}| \, . 
\end{align}

\end{theorem}
The assertion of~\cref{Master} follows from the following two lemmas, since $\gamma_{\LOCC}(U) \leq \gamma_{\LO}(U)$.

\begin{lemma} \label{lem_upper_bound}
Let $U_{AB}$ be of the same form as in~\cref{Master}. It then holds that
\begin{align}
\gamma_{\LO}(U) \leq \sum_k |u_k|^2 + 2 \sum_{i \ne j} |u_i| |u_j| \, . 
\end{align}
\end{lemma}
\begin{proof}
We prove this statement by explicitly constructing a quasiprobability decomposition of $U$ into $\LO$ which exhibits a one-norm of the quasiprobability coefficients given by $\sum_k |u_k|^2 + 2 \sum_{i \ne j} |u_i| |u_j|$.
This proof is an extension of the decomposition basis derived in~\cite{Mitarai_2021}, which is optimal for some two-qubit unitaries but not for all of them. 
As explained in the introduction~\cref{ssec:kak}, we can directly consider the interaction part $W$ of the unitary $U$ (up to local unitaries $U$ and $W$ are equal). The induced channel of $W$ is given by
\begin{equation}
    \cW(\rho_{AB}) = \sum_{k, k'} u_k u_{k'}^* (L_k)_A \otimes (R_k)_B \rho_{AB} (L_{k'})_A^{\dagger} \otimes (R_{k'})_B^{\dagger} \, .
\end{equation}
By writing the complex factors by its argument and absolute value $u_k = |u_k|\ee^{\ci\phi_k} $, we get
\begin{align}
    &\cW(\rho) = \sum_{k} |u_k|^2 (L_k)_A \otimes (R_k)_B  \rho_{AB} (L_k)_A^{\dagger} \otimes (R_k)_B^{\dagger}  \nonumber \\
     &\hspace{13mm}+  \sum_{k < {k'}} |u_k| |u_{k'}| \big(\ee^{\ci\phi_{k,{k'}}}(L_k)_A \otimes (R_k)_B  \rho_{AB} (L_{k'})_A^{\dagger} \otimes (R_{k'})_B^{\dagger} \nonumber\\
     & \hspace{18mm} + \ee^{-\ci\phi_{k,{k'}}} (L_{k'})_A \otimes (R_{k'})_B \rho_{AB} (L_k)_A^{\dagger} \otimes (R_k)_B^{\dagger} \big) \, , \label{eq_step1}
\end{align}
where $\phi_{k,{k'}} = \phi_k - \phi_{k'}$\footnote{Note that in contrast to \cite{Mitarai_2021}, we describe $e^{i\phi_{k,{k'}}}L_k \otimes R_k \rho L_{k'} \ox R_{k'} + e^{-i\phi_{k,{k'}}} L_{k'} \otimes R_{k'} \rho L_k \otimes R_k$ with local operations instead of $\sigma_k \otimes \sigma_k \rho \sigma_{k'} \otimes \sigma_{k'}$. This will allow us to consider the phase of the $u_k$ to find the optimal decomposition.}.
Implementing the first sum with $\LO$ is trivial, and the difficulty will be the implementation of the term:
\begin{equation}
    \ee^{\ci\phi_{k,{k'}}}(L_k)_A \otimes (R_k)_B  \rho_{AB} (L_{k'})_A^{\dagger} \otimes (R_{k'})_B^{\dagger} + \ee^{-\ci\phi_{k,{k'}}} (L_{k'})_A \otimes (R_{k'})_B \rho_{AB} (L_k)_A^{\dagger} \otimes (R_k)_B^{\dagger} \label{eq_expression}
\end{equation}
To do so, we define the expressions 
    \begin{align}
        ({l_{k,k'}^{\phi}})_\pm = \frac{L_{k} \pm \ee^{-\ci \phi} L_{k'} }{2}  \qquad  \textnormal{and} \qquad ({r_{k,k'}^{\phi}})_\pm = \frac{R_{k} \pm \ee^{-\ci \phi} R_{k'}}{2} \, ,
    \end{align}
and based on them the maps
    \begin{align}
        ({\cC_{k,k'}^{\phi}})_\pm(\rho) = ({l_{k,k'}^{\phi}})_\pm \, \rho \,  ({l_{k,k'}^{\phi}})_\pm^{\dagger} \qquad \textnormal{and}  \qquad
        ({\cD_{k,k'}^{\phi}})_\pm(\rho) = ({r_{k,k'}^{\phi}})_\pm \, \rho \, ({r_{k,k'}^{\phi}})_\pm^{\dagger} \, .
    \end{align}
Under usage of the \emph{negativity} trick, we then define the channels
    \begin{align}
        {\cC_{k,k'}^{\phi}}(\rho) = ({\cC_{k,k'}^{\phi}})_+(\rho) -({\cC_{k,k'}^{\phi}})_-(\rho) \quad \textnormal{and} \quad
        {\cD_{k,k'}^{\phi}}(\rho) = ({\cD_{k,k'}^{\phi}})_+(\rho) -({\cD_{k,k'}^{\phi}})_-(\rho)   \, ,
    \end{align}
which evaluate to
    \begin{align}
        {\cC_{k,k'}^{\phi}}(\rho) = \frac{1}{2} \big(\ee^{\ci\phi} L_k \rho L_{k'}^{\dagger}  + \ee^{-\ci\phi} L_{k'} \rho L_k^{\dagger} \big)\quad \textnormal{and} \quad
        {\cD_{k,k'}^{\phi}}(\rho) &= \frac{1}{2} \big(\ee^{\ci\phi} R_k \rho R_{k'}^{\dagger}  + \ee^{-\ci\phi} R_{k'} \rho R_k^{\dagger} \big)  \, .
    \end{align}
This expression already resembles~\cref{eq_expression}, but is so far only a local
channel.\footnote{One can check that $({\mathcal{C}_{k,{k'}}^\phi})_\pm, ({\mathcal{D}_{k,{k'}}^\phi})_\pm \succeq 0$ and $({\mathcal{C}_{k,{k'}}^\phi})_+ +  ({\mathcal{C}_{k,{k'}}^\phi})_- , ({\mathcal{D}_{k,{k'}}^\phi})_+ +  ({\mathcal{D}_{k,{k'}}^\phi})_-\in \CPTP$ are fulfilled.} 
In order to obtain~\cref{eq_expression}, we apply these channels in parallel. This gives us
\begin{align}
({\cC_{k,k'}^{\phi}} \otimes {\cD_{k,k'}^{\phi}}) (\rho) 
&= \frac{1}{4} \ee^{2\ci\phi} (L_k \otimes R_k) \rho (L_{k'}^{\dagger} \otimes R_{k'}^{\dagger}) + \frac{1}{4}\ee^{-2\ci\phi} (L_{k'} \otimes R_{k'}) \rho (L_k^{\dagger} \otimes R_k^{\dagger}) \nonumber \\
&\hspace{5mm}+ \frac{1}{4} (L_k \otimes R_{k'}) \rho (L_{k'}^{\dagger} \otimes R_k^{\dagger}) + \frac{1}{4} (L_{k'} \otimes R_k) \rho (L_k^{\dagger} \otimes R_{k'}^{\dagger}) \, ,
\end{align}
which resembles~\cref{eq_expression} even more, but we still have mixed terms we want to get rid of.
We do so by considering the operation ${\cC_{k,k'}^{\phi+\frac{\pi}{2}}}$ and ${\cD_{k,k'}^{\phi+\frac{\pi}{2}}}$. This introduces a global and a relative phase shift. For instance for ${\cC_{k,k'}^{\phi+\frac{\pi}{2}}}$, we get
\begin{align}
    \mathcal{C}_{k,{k'}}^{\phi+\frac{\pi}{2}}(\rho) 
    = \frac{1}{2} \ee^{\ci\frac{\pi}{2}} (L_k \rho L_{k'}^{\dagger} \ee^{\ci\phi}+ \ee^{-i\pi} L_{k'} \rho L_k^{\dagger} \ee^{-\ci\phi}) 
    = \frac{1}{2} \ee^{\ci\frac{\pi}{2}} (L_k \rho L_{k'}^{\dagger} \ee^{\ci\phi} - L_{k'} \rho L_k^{\dagger} \ee^{-\ci\phi}) \, ,
\end{align}
which extended to both parties yields
\begin{align}
    (\mathcal{C}_{k,{k'}}^{\phi+\frac{\pi}{2}} \otimes \mathcal{D}_{k,{k'}}^{\phi+\frac{\pi}{2}}) (\rho) 
    &= -\frac{1}{4} \ee^{2\ci\phi}(L_k \otimes R_k) \rho (L_{k'}^{\dagger} \otimes R_{k'}^{\dagger}) - \frac{1}{4}\ee^{-2\ci\phi} (L_{k'} \otimes R_{k'}) \rho (L_k^{\dagger} \otimes R_k^{\dagger}) \nonumber \\
     &\hspace{4mm}+ \frac{1}{4} (L_k \otimes R_{k'}) \rho (L_{k'}^{\dagger} \otimes R_k^{\dagger}) + \frac{1}{4} (L_{k'} \otimes R_k) \rho (L_k^{\dagger} \otimes R_{k'}^{\dagger}) \, .
\end{align}
This allows us to formulate the desired term as
\begin{align}
   & 2 \left( \cC_{k,{k'}}^{\frac{\phi_{k,{k'}}}{2}}\otimes\cD_{k,{k'}}^{\frac{\phi_{k,{k'}}}{2}}  - \cC_{k,{k'}}^{\frac{\phi_{k,{k'}}}{2}+\frac{\pi}{2}} \ox \cD_{k,{k'}}^{\frac{\phi_{k,{k'}}}{2}+\frac{\pi}{2}} \right) \nonumber \\
   &\hspace{20mm}= 
    \ee^{\ci\phi_{k,{k'}}}(L_k \otimes R_k) \rho (L_{k'}^{\dagger} \otimes R_{k'}^{\dagger}) + \ee^{-\ci\phi_{k,{k'}}} (L_{k'} \otimes R_{k'}) \rho (L_k^{\dagger} \otimes R_k^{\dagger})  \, .
\end{align}
Combining this with~\cref{eq_step1} gives
\begin{align}
    \mathcal{W}(\rho) 
    &= \sum_{k} |u_k|^2 (L_k)_A \otimes (R_k)_B  \rho_{AB} (L_k)_A^{\dagger} \otimes (R_k)_B^{\dagger}  \nonumber \\
    &+ \sum_{k<{k'}} |u_k||u_{k'}| \, 2  \left( \cC_{k,{k'}}^{\frac{\phi_{k,{k'}}}{2}}\otimes\cD_{k,{k'}}^{\frac{\phi_{k,{k'}}}{2}}  - \cC_{k,{k'}}^{\frac{\phi_{k,{k'}}}{2}+\frac{\pi}{2}} \ox \cD_{k,{k'}}^{\frac{\phi_{k,{k'}}}{2}+\frac{\pi}{2}} \right) \, ,
\end{align}
which implies
\begin{equation}
    \gamma_{\LO}(U_{AB}) 
    \leq \sum_{k} |u_k|^2 + 4 \sum_{k<{k'}} |u_k||u_{k'}|
    = \sum_{k} |u_k|^2 + 2 \sum_{k \ne {k'}} |u_k||u_{k'}| \, .
\end{equation}
\end{proof}

This construction that handles the interference terms is also known as Hadamard test and is used with a similar aim in~\cite{BSS16}. It differs however from our manuscript, in that it decomposes the unitary in the Pauli basis which is not necessarily the best decomposition. Therefore it only provides an upper bound on the optimal simulation overhead, which is not tight for most instances.

For convenience, we present the circuits that implement the channels $({\cC_{k,k'}^{\phi}})_\pm(\rho)$ and $({\cD_{k,k'}^{\phi}})_\pm(\rho)$ in \cref{sec:circ} and explain how to apply them.

\begin{lemma} \label{lem_lower_bound}
Let $U_{AB}$ be of the form described in~\cref{Master} and furthermore require orthogonality of $L_k$ and $R_k$ in the following sense:
\begin{align}
    \tr(L_k^\dagger L_{k'}) = \dim(A) \delta_{k,k'} \qquad \textnormal{and} \qquad 
    \tr(R_k^\dagger R_{k'}) = \dim(B) \delta_{k,k'} \, .
\end{align}
Then
\begin{align}
\gamma_{\LOCC}(U) \geq 1 + 2 \sum_{i \ne j} |u_k| |u_{k'}| \, . 
\end{align}
\end{lemma}
\begin{proof}
Let $\Choi$ denote the Choi state of a unitary $U_{AB}$.\footnote{We define the Choi state of $U_A$ as $\Choi = \id_{A'} \ox U_A \ket{\psi}_{A'A}$, where $A'$ is a copy of the Hilbert space $A$ and $\ket{\psi}_{A'A}$ is the maximally entangled state on the joint system $A'A$.}
Simulating the channel $U_{AB}$ is at least as expensive as preparing a state that can be prepared using the channel --- in particular the Choi state. We therefore have
\begin{equation}\label{eq_gammaState}
    \gamma_\LOCC(U_{AB}) = \gamma_\LOCC(W_{AB})  \geq \gamma_\LOCC(\ket{\Phi_W}) \,, 
\end{equation}
where we introduced the $\gamma$-factor for states and $W_{AB}$ denotes the interaction unitary. Furthermore, the Choi state of $W_{AB}$ is denoted by $\ket{\Phi_W}$. 
The $\gamma$-factor for states is defined by
\begin{align}
      \gamma_S(\rho_{AB}) \coloneqq \min \Big\{ \gamma_S(\cE) | \cE(\proj{0}_{AB}) = \rho_{AB} \Big\} \, ,
\end{align}
but as was shown in~\cite[Lemma 4.1]{piv23}, for pure states it can be written as a function of the Schmidt coefficients. One finds
\begin{equation}
    \gamma_\LOCC(\ket{\Phi_W}) = 2\left(\sum_i \alpha_i \right)^2-1 \,, 
\end{equation}
where the $\alpha_i$ are the Schmidt coefficients of $\ket{\Phi_W}$.
We can calculate these concretely in terms of $u_i$. The Choi state is then given as
\begin{align}
        \ket{\Phi_W}_{A'B'AB} &= \frac{1}{\sqrt{\dim(A)\dim(B)}} \sum_{i=1}^{\dim(A)} \sum_{j=1}^{\dim(B)} \ket{i,j}_{A'B'} \ox W\ket{i,j}_{AB} \\
        & = \frac{1}{\sqrt{\dim(A)\dim(B)}}  \sum_{i=1}^{\dim(A)} \sum_{j=1}^{\dim(B)} \sum_k u_{k}  \ket{i,j}_{A'B'} \ox (L_k\ox R_k) \ket{i,j}_{AB} \\
        &= \sum_{k} |u_{k}| \ket{\varphi_{k}^{(1)}}_{A'A} \ox \ket{\varphi_{k}^{(2)}}_{B'B} \, ,
    \end{align}
where $\ket{\varphi_k}_{A'A}$ and $\ket{\varphi_k}_{B'B}$ are orthonormal vectors given by
\begin{align}
    \ket{\varphi_k^{(1)}}_{A'A} &= \frac{1}{\sqrt{\dim(A)}}\sum_{j=1}^{\dim(A)}\ket{j}_{A'}L_k\ket{j}_A \qquad \textnormal{and} \\
    \ket{\varphi_k^{(2)}}_{B'B} &= \frac{1}{\sqrt{\dim(B)}}\arg(u_k) \sum_{j=1}^{\dim(B)}  \ket{j}_{B'}  R_k \ket{j}_{B} \, .
\end{align}
Orthonormality follows from the above trace properties of $L_k$ and $R_k$.
In this form, the Choi state is in its Schmidt decomposition and we can identify the Schmidt coefficients $\alpha_i$ as $|u_i|$. Using \cref{eq_gammaState} and an analogous calculation to the proof of \cref{co_bounded}, we get
\begin{equation}
    \gamma_\LOCC \geq 2\left(\sum_{i=0}^3 |u_i|\right)^2-1 = 1 + 2\sum_{i\ne j} |u_i||u_j| \, ,
\end{equation}
which concludes the proof.
\end{proof}

As an additional application of~\cref{Master} and our complete understanding of the parallel cut setting, we show that one can leverage this theorem to explicitly find decompositions of unitaries that consist of multiple two-qubit unitaries interleaved by arbitrary, but known\footnote{Note that we are not referring to the black box setting here.} operations as depicted in~\cref{fig:cut}.
In general, we will not be able to give the optimal decomposition for these scenarios as this would rely on our ability to cut arbitrary unitaries, however we can still provide a (not necessarily optimal) decomposition with an overhead that is given by the optimal decomposition of the parallel cut. This decomposition will depend on the interleaved operations.
\begin{figure}[!htb]
    \centering
        \begin{tikzpicture}[thick,scale=0.8]
    \def\x{0}
    \def\b{7}
    \def \s{2.5}

     \draw [fill=gray!15,draw=none] (\x+0.05,0.06) rectangle (\x+\b,1.9);   
     \draw [fill=cyan!15,draw=none] (\x+0.05,-0.06) rectangle (\x+\b,-1.9);

     \draw [fill=gray!15,draw=none] (\s+0.1,1.4+0.8) rectangle (\s+0.1+0.25,1.4+0.25+0.8);
     \draw [fill=cyan!15,draw=none] (\s+1.6,1.4+0.8) rectangle (\s+1.6+0.25,1.4+0.25+0.8);     
     \node[gray] at (\s + 0.6,1.9+0.4) {\footnotesize{$A$}};
     \node[cyan] at (\s + 2.1,1.9+0.4) {\footnotesize{$B$}};

\node at (\x+\b/2,0) {$-----------------$};


     \draw (\x,0.5) -- (\x+\b+0.05,0.5);
     \draw (\x,-0.5) -- (\x+\b+0.05,-0.5); 
     
     \draw (\x,0.9) -- (\x+\b+0.05,0.9);
     \draw (\x,-0.9) -- (\x+\b+0.05,-0.9);

     \draw (\x,1.3) -- (\x+\b+0.05,1.3);
     \draw (\x,-1.3) -- (\x+\b+0.05,-1.3);

     \draw (\x,1.7) -- (\x+\b+0.05,1.7);
     \draw (\x,-1.7) -- (\x+\b+0.05,-1.7);

     \def \xx{0.5}
     \draw [fill=yellow!50] (\xx-0.36,1.85) rectangle (\xx+0.36,0.1); 
     \node at (\xx+0.03,0.9) {\scriptsize{$V_A^{(3)}$}};
     \draw [fill=yellow!50] (\xx-0.36,-0.1) rectangle (\xx+0.36,-1.85); 
     \node at (\xx,-0.9) {\scriptsize{$W_B^{(3)}$}};

     \def \xx{1.5+\x}
     \draw [fill=red!40] (\xx,0.9) circle (3mm);
     \node at (\xx+0.02,0.9) {\scriptsize{$U_3$}};
     \draw (\xx,0.9-0.3) -- (\xx,-0.5+0.3);
     \draw [fill=red!40](\xx,-0.5) circle (3mm);
     \node at (\xx+0.02,-0.5) {\scriptsize{$U_3$}};

     \def \xx{2.5}
     \draw [fill=yellow!50] (\xx-0.36,1.85) rectangle (\xx+0.36,0.1); 
     \node at (\xx+0.03,0.9) {\scriptsize{$V_A^{(2)}$}};
     \draw [fill=yellow!50] (\xx-0.36,-0.1) rectangle (\xx+0.36,-1.85); 
     \node at (\xx+0.01,-0.9) {\scriptsize{$W_B^{(2)}$}};

     \def \xx{3.5}
     \draw [fill=ForestGreen!40] (\xx,0.5) circle (3mm);
     \node at (\xx+0.02,0.5) {\scriptsize{$U_2$}};
     \draw (\xx,0.5-0.3) -- (\xx,-0.9+0.3);
     \draw [fill=ForestGreen!40](\xx,-0.9) circle (3mm);
     \node at (\xx+0.02,-0.9) {\scriptsize{$U_2$}};

     \def \xx{4.5}
     \draw [fill=yellow!50] (\xx-0.36,1.85) rectangle (\xx+0.36,0.1); 
     \node at (\xx+0.03,0.9) {\scriptsize{$V_A^{(1)}$}};
     \draw [fill=yellow!50] (\xx-0.36,-0.1) rectangle (\xx+0.36,-1.85); 
     \node at (\xx+0.01,-0.9) {\scriptsize{$W_B^{(1)}$}};

     \def \xx{5.5}
     \draw [fill=Violet!40] (\xx,1.3) circle (3mm);
     \node at (\xx+0.02,1.3) {\scriptsize{$U_1$}};
     \draw (\xx,1.3-0.3) -- (\xx,-1.3+0.3);
     \draw [fill=Violet!40](\xx,-1.3) circle (3mm);
     \node at (\xx+0.02,-1.3) {\scriptsize{$U_1$}};
    
     \def \xx{6.5}
     \draw [fill=yellow!50] (\xx-0.36,1.85) rectangle (\xx+0.36,0.1); 
     \node at (\xx+0.03,0.9) {\scriptsize{$V_A^{(0)}$}};
     \draw [fill=yellow!50] (\xx-0.36,-0.1) rectangle (\xx+0.36,-1.85); 
     \node at (\xx+0.01,-0.9) {\scriptsize{$W_B^{(0)}$}};

    \end{tikzpicture}
    \caption{Cut of multiple two-qubit unitaries (here denoted by $U_1,U_2,U_3$) with interleaved operations (denoted by $V_A^{(i)},W_B^{(j)}$).} 
    \label{fig:cut}
\end{figure}
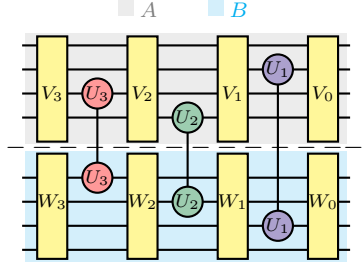

\begin{corollary}[Parallel cut as upper bound]\label{theorem3}
    Let $U$ be a unitary of the form 
    \begin{equation}
        U = (V_A^{(0)}\ox W_B^{(0)}) \prod_{i=1}^n \left(U_{A_i B_i}^{(i)} (V_A^{(i)}\ox W_B^{(i)}) \right) \, ,
    \end{equation}
    where $V^{(i)}_A$ and $W^{(i)}_B$ are arbitrary unitaries on the subsystems $A$ and $B$, respectively and $U_{A_i B_i}$ acts as a two-qubit unitary between $A$ and $B$.
    We then have
    \begin{equation}\label{eq_parallel-as-upper}
        \gamma_\LO(U) \leq \gamma(\bigotimes_{i=1}^n U_i) \, .
    \end{equation}
\end{corollary}
Since cutting the gates $U_i$ jointly is strictly cheaper than cutting them individually, cutting the whole circuit $U$ is also cheaper than cutting the $U_i$ individually.

\begin{proof}
    We show the upper bound by constructing an explicit decomposition with overhead equal to the parallel cut scenario. 
    Take a unitary $U$ of the form of \cref{theorem3} with two-qubit unitaries $U_{A_i B_i}^{(i)}$. To simplify the notation, we will assume that \smash{$U_{A_i B_i}^{(i)}$} acts on the $i^{th}$ qubit of $A$ and $B$.
    We then perform a KAK decomposition of all $U_{A_i B_i}^{(i)}$, i.e.
    \begin{align}
        U_{A_i B_i}^{(i)} = \left(V_{A_i}^{(1,i)}\ox W_{B_i}^{(2,i)} \right)  \left( \sum_{k=0}^3 u_k^{(i)} (\sigma_k)_{A_i} \ox (\sigma_k)_{B_i} \right)  \left(V_{A_i}^{(3,i)} \ox W_{B_i}^{(4,i)} \right) \,,
    \end{align}
    and absorb the local one-qubit unitaries $V_{A_i}^{(1,i)}, V_{A_i}^{(2,i)}, W_{B_i}^{(1,i)}, W_{B_i}^{(2,i)},$  into the unitaries \smash{$V_A^{(i-1)},W_B^{(i-1)}$} and \smash{$V_A^{(i)},W_B^{(i)}$}. This yields
    \begin{align}
        U &= \sum_{k_1,\ldots,k_n=0}^3 (V_A^{(0)}\ox W_B^{(0)}) \prod_{i=1}^n u^{(i)}_{k_i}   \left({(\sigma_{k_i}})_{A_i}\ox {(\sigma_{k_i}})_{B_i}(V_A^{(i)}\ox W_B^{(i)}) \right) \\
        &= \sum_{k_1,\ldots,k_n=0}^3 \prod_{i=0}^n u^{(i)}_{k_i}  ({L_{k_1,\dots,k_n}})_A \otimes ({R_{k_1,\dots,k_n}})_B \\
        &= \sum_{\mathbf{k}}  u_{\mathbf{k}} \, \,  ({L_{\mathbf{k}}})_A \otimes ({R_{\mathbf{k}}})_B \, ,
    \end{align}
    where we introduced $u_{\mathbf{k}} = \prod_{i=0}^n u^{(i)}_{k_i}$ and the unitaries 
    \begin{align}
        ({L_{\mathbf{k}}})_A &= ({L_{k_1,\dots,k_n}})_A = V_A^{(0)} \prod_{i=1}^n {(\sigma_{k_i}})_{A_i} V_A^{(i)}   \\
        ({R_{\mathbf{k}}})_B &= ({R_{k_1,\dots,k_n}})_B = W_B^{(0)} \prod_{i=1}^n {(\sigma_{k_i}})_{B_i} V_B^{(i)}  \, . 
    \end{align}
    In this form, we see that the first part of~\cref{Master} applies. Since $V_A^{(i)}$ and $V_B^{(i)}$ are arbitrary, the unitaries $L_{\mathbf{k}}$ and $R_{\mathbf{k}}$ are in general not orthogonal and therefore optimality is not necessarily achieved. However, since $u_{\mathbf{k}}$ is a product of KAK coefficients, we still find $\sum_{\mathbf{k}}  |u_{\mathbf{k}}|^2 = 1$.
\end{proof}
\begin{remark} \label{rmk_long}
The following comments concerning~\Cref{theorem3} are worth being mentioned.
\begin{enumerate}[(i)]
\item \cref{theorem3} shows that performing local operations between two-qubit gates only reduces the overhead of the unitary compared to the parallel cut case (\smash{$V_A^{(i)}=W_B^{(i)}=\mathds{1}$} for all $i$). Intuitively, we can understand this since the case that all two-qubit gates act on different qubits is already the case where we create the most entanglement between both partitions $A$ and $B$. Local unitaries cannot increase this entanglement. However, they can reduce it. Consider the case where \smash{$V_A^{(i)}, W_B^{(i)}$} are SWAP gates that make all two-qubit unitaries unitaries between $A$ and $B$ effectively act on the same two qubits.
In this case the overall circuit is equivalent to a single two-qubit unitary and cutting it costs a maximum of 7 and therefore considerably less than cutting all two-qubit gates simultaneously.
\item  The channel \smash{$({\cC_{\mathbf{k},\mathbf{k'}}^{\phi}}\ox{\cD_{\mathbf{k},\mathbf{k'}}^{\phi}})(\rho)$} depends on the interleaved operations \smash{$V_A^{(i)}, W_B^{(i)}$}. If the gates \smash{$V_A^{(i)}, W_B^{(i)}$} are simple, the calculation of the channel decomposition is feasible. However, if the interleaved gates are complicated, e.g.~in the situation where one two-qubit gate at the beginning and one at the end of a long circuit are to be cut simultaneously, finding the $\LO$ decomposition might be difficult. \label{rem_hard}
\end{enumerate}
\end{remark}


\section{Cutting in the black box setting}\label{sec_black_box}
As discussed in~\cref{rmk_long} (\cref{rem_hard}), evaluating the optimal decomposition  of a unitary consisting of two-qubit gates interleaved with arbitrary unitaries might be difficult if the evaluation of the interleaved elements is complicated. Therefore this section will focus on the construction of a decomposition that is independent of the form of the interleaved element. This is what we call the black box cut as shown in~\cref{fig_blackBoxCut}.
We believe that this scenario is also the most useful in practice as it allows us to jointly cut multiple two-qubit gates, which reduces the sampling overhead, without making any assumptions on where in the circuit these gates may lie.

Before we present the theorem that covers this case, let us define the minimal overhead in the black box setting. For simplicity, we will first restrict to the case of two two-qubit gates, but the case with arbitrary number of gates follows analogously.
\begin{definition}[Black box $\gamma$-factor] \label{def_blackbox_factor}
    We define the minimal overhead of two channels that are separated by an arbitrary black box channel using operations $S(A,B)=\{ \LO(A,B),\LOCC(A,B) \}$ as 
  \begin{align}
    \overline{\gamma}_S(\cE_{AB}^1 ,\cE_{AB}^2 ) 
      &\coloneqq \min \Big\{ \sum_{i=1}^m \lvert a_i \rvert :  \cE_{AB}^1 X_{AB} \cE_{AB}^2 = \tr_{E_AE_B} \sum\limits_{i=1}^m a_i \cF_i X_{AB} \cG_i, \nonumber \\ 
      &\hspace{25mm} \forall X_{AB}\in\CPTP, \, m\geq 1, \, \cF_i,\cG_i\in S(AE_A,BE_B) \textnormal{ and } a_i \in \R \Big\} \, .
  \end{align} 
  Here, we allow the introduction of an ancillary system $E_A$ and $E_B$ on which the black box $X$ is not acting. This allows to correlate channels before and after the black box computation.
\end{definition}

We can then formulate our theorem.
\begin{theorem}\label{theorem:blackbox}
    Two two-qubit unitaries $U_{A_1,B_1}$ and $V_{A_2,B_2}$ have a decomposition independent of a black box channel, with optimal black box overhead
    \begin{equation}
        \overline{\gamma}_{\LOCC}(U_{A_1,B_1}, V_{A_2,B_2}) = \overline{\gamma}_{\LO}(U_{A_1,B_1}, V_{A_2,B_2}) = 1 + 2 \sum_{\mathbf{k} \neq\mathbf{k}'} |w_{\mathbf{k}}| |w_{\mathbf{k}'}| \, ,
    \end{equation}
    where $\mathbf{k} \in \{0,1,2,3\}^2$ and $w_{\mathbf{k}} = u_{k_1} v_{k_2}$ is the product of the KAK-coefficients of $U_{A_1,B_1}$ and $V_{A_2,B_2}$.
\end{theorem}

To prevent expression from becoming too convoluted, we will introduce a graphical notations of tensor products for the following section.
In this notation we always consider channels which act from both sides on a density matrix. These expression have to be read from the inside, usually $\rho$, to the outside. The vertical dimension denotes that operations act on different subsystems.
\begin{equation}
     \begin{tikzpicture}[thick,scale=0.8]
    \centering
    \def \x{6}
    \def \s{0}

    \node at (0,0) {$\quad \quad \sum_{i,j} (\sigma_i)_A \ox (\sigma_i)_B \rho_{AB} (\sigma_j)_A \ox (\sigma_j)_B
    \overset{\textnormal{notation}}{=} \sum_{i,j}$};

    \def \s{\x-0.7}
    \draw (\s,0.3) -- (\x+0.8,0.3);
    \draw (\s,-0.3) -- (\x+0.8,-0.3);
    
    \draw [fill=gray!15] (\s, 0.5) rectangle (\s+0.4, 0.1);
    \node at (\s+0.2,0.25) {\scriptsize$\sigma_i$};
    
    \draw [fill=cyan!15] (\s, -0.5) rectangle (\s+0.4, -0.1);
    \node at (\s+0.2,-0.3) {\scriptsize$\sigma_i$};

    \draw [fill=white!15] (\x , 0.5) rectangle (\x+0.5,-0.5);
    \node at (\x+0.25,0) {$\rho$};

    \def \s{\x+0.8}
    \draw [fill=gray!15] (\s, 0.5) rectangle (\s+0.4, 0.1);
    \node at (\s+0.2,0.25) {\scriptsize$\sigma_j$};
    
    \draw [fill=cyan!15] (\s, -0.5) rectangle (\s+0.4, -0.1);
    \node at (\s+0.2,-0.3) {\scriptsize$\sigma_j$};

    \end{tikzpicture}   \, .
\end{equation}

\begin{proof}
    The essential ingredient for achieving submultiplicativity is to correlate the channels before and after the black box. We achieve this, by reducing the problem to the parallel cut and by making use of ancillary qubits. These are used to perform Bell measurements similar to the gate teleportation protocol.
    We will illustrate the procedure using the above introduced notation.
    The overall channel we want to implement is given by
    \begin{equation}
            \begin{tikzpicture}[thick,scale=0.8]
    \centering
    \def \x{7}
    \def \s{0}

    \node at (0,0) {$\quad \quad \cE(\rho)  = \cU ( \cX ( \cV (\rho) ) ) = \sum_{k,k',i,l,l'}  u_{k}u_{k'}^*  v_l v_{l'}
    \,$};

    \draw (\x-2,0.3) -- (\x+2.5,0.3);
    \draw (\x-2,-0.3) -- (\x+2.5,-0.3);

    \def \s{\x-2.1}
    
    \draw [fill=gray!15] (\s, 0.5) rectangle (\s+0.5, 0.1);
    \node at (\s+0.25,0.25) {\scriptsize$\sigma_k$};
    
    \draw [fill=cyan!15] (\s, -0.5) rectangle (\s+0.5, -0.1);
    \node at (\s+0.25,-0.3) {\scriptsize$\sigma_k$};

    \def \s{\x-1.4}
    \draw [fill=black!90] (\s , 0.5) rectangle (\s+0.5,-0.5);
    \node[white] at (\s+0.25,0) {$X_i$};

    \def \s{\x-0.7}
    
    \draw [fill=gray!15] (\s, 0.5) rectangle (\s+0.5, 0.1);
    \node at (\s+0.25,0.25) {\scriptsize$\sigma_l$};
    
    \draw [fill=cyan!15] (\s, -0.5) rectangle (\s+0.5, -0.1);
    \node at (\s+0.25,-0.3) {\scriptsize$\sigma_l$};

    \draw [fill=white!15] (\x , 0.5) rectangle (\x+0.5,-0.5);
    \node at (\x+0.25,0) {$\rho$};

    \def \s{\x+0.8}
    \draw [fill=gray!15] (\s, 0.5) rectangle (\s+0.5, 0.1);
    \node at (\s+0.28,0.25) {\scriptsize$\sigma_{l'}$};
    
    \draw [fill=cyan!15] (\s, -0.5) rectangle (\s+0.5, -0.1);
    \node at (\s+0.28,-0.3) {\scriptsize$\sigma_{l'}$};

    \def \s{\x+1.5}
    \draw [fill=black!90] (\s , 0.5) rectangle (\s+0.5,-0.5);
    \node[white] at (\s+0.25,0) {$X_i^\dagger$};

    \def \s{\x+2.2}
    
    \draw [fill=gray!15] (\s, 0.5) rectangle (\s+0.5, 0.1);
    \node at (\s+0.28,0.25) {\scriptsize$\sigma_{k'}$};
    
    \draw [fill=cyan!15] (\s, -0.5) rectangle (\s+0.5, -0.1);
    \node at (\s+0.28,-0.3) {\scriptsize$\sigma_{k'}$};

    \end{tikzpicture} \, ,
    \end{equation}
    
    where we introduce Kraus operators $X_i$ for the black box channel $\cX$.
    To correlate $\cU$ and $\cV$, we introduce additional qubits in the Bell-state \emph{$\ket{\phi_0} = \frac{1}{\sqrt{2}} (\ket{00} + \ket{11}),\, (\phi_0 = \ket{\phi_0}\bra{\phi_0})\,$} for both parties. Then we perform the channel we used for parallel two-qubit gates (cf. \cref{cor_parallel_cuts}) on system and ancillary qubits. 
\begin{equation}
        \begin{tikzpicture}[thick,scale=0.8]
    \centering
    \def \x{5}
    \def \s{0}
    \def \c{1}

    \node at (0,0) {$\quad \quad \cG(\rho) = \sum_{k,l,k',l'}  u_{k}u_{k'}^* v_{l}v_{l'}^* 
    \,$};
    
    \def \a{\x-1.2*\c}
    \def \b{\x+2.2*\c}
    \draw (\a,\c/2) -- (\b,\c/2);
    \draw (\a,-\c/2) -- (\b,-\c/2);
    \draw (\a,3*\c/2) -- (\b,3*\c/2);
    \draw (\a,-\c-\c/2) -- (\b,-\c-\c/2);
    \draw (\a,2*\c+\c/2) -- (\b,2*\c+\c/2);
    \draw (\a,-2*\c-\c/2) -- (\b,-2*\c-\c/2);

    \def \s{\x-1.2*\c}
    
    \draw [fill=gray!15] (\s, 3*\c) rectangle (\s+\c, 2*\c+0.1);
    \node at (\s+\c/2,5*\c/2) {$\sigma_{k}$};
    
    \draw [fill=gray!15] (\s, 1*\c) rectangle (\s+\c, 0.1);
    \node at (\s+\c/2,\c/2) {$\sigma_l$};
    
    \draw [fill=cyan!15] (\s, -1*\c) rectangle (\s+\c, -0.1);
    \node at (\s+\c/2,-\c/2) {$\sigma_l$};

    \draw [fill=cyan!15] (\s, -3*\c) rectangle (\s+\c, -2*\c-0.1);
    \node at (\s+\c/2,-2*\c-\c/2) {$\sigma_{k}$};


    \draw [fill=white!15] (\x , \c) rectangle (\x+\c,-\c);
    \node at (\x+\c/2,0) {$\Tilde{\rho}$};

    \draw [fill=gray!15] (\x, 3*\c) rectangle (\x+\c, \c+0.1);
    \node at (\x+\c/2,2*\c+0.1) {$\phi_0$};

    \draw [fill=cyan!15] (\x, -3*\c) rectangle (\x+\c, -\c-0.1);
    \node at (\x+\c/2,-2*\c) {$\phi_0$};

    \def \s{\x+1.2*\c}

     \draw [fill=gray!15] (\s, 3*\c) rectangle (\s+\c, 2*\c+0.1);
    \node at (\s+\c/2,5*\c/2) {$\sigma_{k'}$};
    
    \draw [fill=gray!15] (\s, 1*\c) rectangle (\s+\c, 0.1);
    \node at (\s+\c/2,\c/2) {$\sigma_{l'}$};
    
    \draw [fill=cyan!15] (\s, -1*\c) rectangle (\s+\c, -0.1);
    \node at (\s+\c/2,-\c/2) {$\sigma_{l'}$};

    \draw [fill=cyan!15] (\s, -3*\c) rectangle (\s+\c, -2*\c-0.1);
    \node at (\s+\c/2,-2*\c-\c/2) {$\sigma_{k'}$};

    \end{tikzpicture}
\end{equation}
Applying this channel comes at a sampling overhead of $1 + 2 \sum_{\mathbf{k} \neq\mathbf{k}'} |w_{\mathbf{k}}| |w_{\mathbf{k}'}|$ and requires only local operations.
After that, we can apply any black box channel $\cX$ as long as it does not act on the ancillary qubits. We will call the so obtained state $\Tilde{\rho}$. In order to teleport the gate from the quasi Choi-state, we perform a measurement in the Bell-basis between the unused part of the ancillary qubit and the qubit on which we want the second unitary to act.
We will illustrate the measurement in the Bell-basis as $\bra{\phi_a}$, where the index $a$ corresponds to the measurement outcome.
The post-measurement state corresponding to the measurement outcomes $a$ and $b$ is then given by
\begin{equation}
        \begin{tikzpicture}[thick,scale=0.8]
    \centering
    \def \x{3}
    \def \s{0}
    \def \c{1}

    \node at (0,0) {$\quad \quad \sum_{k,k'}  u_{k}u_{k'}^* 
    \,$};
    
    \def \a{\x-1.2*\c}
    \def \b{\x+2.2*\c}
    \draw (\a,\c/2) -- (\b,\c/2);
    \draw (\a,-\c/2) -- (\b,-\c/2);
    \draw (\a,3*\c/2) -- (\b,3*\c/2);
    \draw (\a,-\c-\c/2) -- (\b,-\c-\c/2);
    \draw (\a,2*\c+\c/2) -- (\b,2*\c+\c/2);
    \draw (\a,-2*\c-\c/2) -- (\b,-2*\c-\c/2);

    \def \s{\x-1.2*\c}
    
    \draw [fill=gray!15] (\s, 3*\c) rectangle (\s+\c, 2*\c+0.1);
    \node at (\s+\c/2,5*\c/2) {$\sigma_{k}$};
    
    \draw [fill=gray!15] (\s, 2*\c) rectangle (\s+\c, 0.1);
    \node at (\s+\c/2,\c+0.1) {$\bra{\phi_a}$};
    
    \draw [fill=cyan!15] (\s, -2*\c) rectangle (\s+\c, -0.1);
    \node at (\s+\c/2,-\c) {$\bra{\phi_b}$};

    \draw [fill=cyan!15] (\s, -3*\c) rectangle (\s+\c, -2*\c-0.1);
    \node at (\s+\c/2,-2*\c-\c/2) {$\sigma_{k}$};


    \draw [fill=white!15] (\x , \c) rectangle (\x+\c,-\c);
    \node at (\x+\c/2,0) {$\Tilde{\rho}$};

    \draw [fill=gray!15] (\x, 3*\c) rectangle (\x+\c, \c+0.1);
    \node at (\x+\c/2,2*\c+0.1) {$\phi_0$};

    \draw [fill=cyan!15] (\x, -3*\c) rectangle (\x+\c, -\c-0.1);
    \node at (\x+\c/2,-2*\c) {$\phi_0$};

    \def \s{\x+1.2*\c}

     \draw [fill=gray!15] (\s, 3*\c) rectangle (\s+\c, 2*\c+0.1);
    \node at (\s+\c/2,5*\c/2) {$\sigma_{k'}$};
    
    \draw [fill=gray!15] (\s, 2*\c) rectangle (\s+\c, 0.1);
    \node at (\s+\c/2,\c+0.1) {$\ket{\phi_a}$};
    
    \draw [fill=cyan!15] (\s, -2*\c) rectangle (\s+\c, -0.1);
    \node at (\s+\c/2,-\c) {$\ket{\phi_b}$};

    \draw [fill=cyan!15] (\s, -3*\c) rectangle (\s+\c, -2*\c-0.1);
    \node at (\s+\c/2,-2*\c-\c/2) {$\sigma_{k'}$};


    \node at (\s+3,0) {$\quad \quad = \quad \sum_{k,k'}  u_{k}u_{k'}^* 
    \,$};

    \def \x{12}

    \def \a{\x-2.4*\c}
    \def \b{\x+2.5*\c}
    \draw (\a,\c/2) -- (\b,\c/2);
    \draw (\a,-\c/2) -- (\b,-\c/2);

    \def \s{\x-2.4*\c}

    \draw [fill=gray!15] (\s, \c) rectangle (\s+\c, 0.1);
    \node at (\s+\c/2,\c/2) {$\sigma_k$};
    
    \draw [fill=cyan!15] (\s, -\c) rectangle (\s+\c, -0.1);
    \node at (\s+\c/2,-\c/2) {$\sigma_k$};
    
    \def \s{\x-1.2*\c}

    \draw [fill=gray!15] (\s, \c) rectangle (\s+\c, 0.1);
    \node at (\s+\c/2,\c/2) {$\sigma_a$};
    
    \draw [fill=cyan!15] (\s, -\c) rectangle (\s+\c, -0.1);
    \node at (\s+\c/2,-\c/2) {$\sigma_b$};

    \draw [fill=white!15] (\x , \c) rectangle (\x+\c,-\c);
    \node at (\x+\c/2,0) {$\Tilde{\rho}$};

    \def \s{\x+1.2*\c}

    \draw [fill=gray!15] (\s, \c) rectangle (\s+\c, 0.1);
    \node at (\s+\c/2,\c/2) {$\sigma_a$};
    
    \draw [fill=cyan!15] (\s, -\c) rectangle (\s+\c, -0.1);
    \node at (\s+\c/2,-\c/2) {$\sigma_b$};

    \def \s{\x+2.4*\c}

    \draw [fill=gray!15] (\s, \c) rectangle (\s+\c, 0.1);
    \node at (\s+\c/2,\c/2) {$\sigma_{k'}$};
    
    \draw [fill=cyan!15] (\s, -\c) rectangle (\s+\c, -0.1);
    \node at (\s+\c/2,-\c/2) {$\sigma_{k'}$};

    \end{tikzpicture}
\end{equation}
where the Pauli operations $\sigma_a$ and $\sigma_b$ appear as a consequence of projecting into the Bell basis.
In the normal gate-teleportation scenario, one can only recover from these Pauli gates, if the teleported gate was a Clifford gate. Even then, one needs to classical communicate the local measurement of $a$ or $b$ to the other side to perform the recovery operation $U(\sigma_a \ox \sigma_b) U^\dagger$. This operation is only local for Clifford gates. The quasiprobability simulation setting is however different. Here, we are not actually performing the non-local unitary, but rather we just sample from local channels. This gives us additional information to work with.

We recover the original channel by first applying a local operation $\sigma_x$ based on the measurement result $x$ of each party. This modifies the post-measurement state to
\begin{equation}
        \begin{tikzpicture}[thick,scale=0.8]
    \centering
    \def \x{5.5}
    \def \s{0}
    \def \c{1}


    \node at (0,0) {$\quad \quad  \sum_{k,k'}  u_{k}u_{k'}^* 
    \,$};

    \def \a{\x-3.6*\c}
    \def \b{\x+3.6*\c}
    \draw (\a,\c/2) -- (\b,\c/2);
    \draw (\a,-\c/2) -- (\b,-\c/2);

    \def \s{\x-3.6*\c}

    \draw [fill=gray!15] (\s, \c) rectangle (\s+\c, 0.1);
    \node at (\s+\c/2,\c/2) {$\sigma_a$};
    
    \draw [fill=cyan!15] (\s, -\c) rectangle (\s+\c, -0.1);
    \node at (\s+\c/2,-\c/2) {$\sigma_b$};

    \def \s{\x-2.4*\c}

    \draw [fill=gray!15] (\s, \c) rectangle (\s+\c, 0.1);
    \node at (\s+\c/2,\c/2) {$\sigma_k$};
    
    \draw [fill=cyan!15] (\s, -\c) rectangle (\s+\c, -0.1);
    \node at (\s+\c/2,-\c/2) {$\sigma_k$};
    
    \def \s{\x-1.2*\c}

    \draw [fill=gray!15] (\s, \c) rectangle (\s+\c, 0.1);
    \node at (\s+\c/2,\c/2) {$\sigma_a$};
    
    \draw [fill=cyan!15] (\s, -\c) rectangle (\s+\c, -0.1);
    \node at (\s+\c/2,-\c/2) {$\sigma_b$};

    \draw [fill=white!15] (\x , \c) rectangle (\x+\c,-\c);
    \node at (\x+\c/2,0) {$\Tilde{\rho}$};

    \def \s{\x+1.2*\c}

    \draw [fill=gray!15] (\s, \c) rectangle (\s+\c, 0.1);
    \node at (\s+\c/2,\c/2) {$\sigma_a$};
    
    \draw [fill=cyan!15] (\s, -\c) rectangle (\s+\c, -0.1);
    \node at (\s+\c/2,-\c/2) {$\sigma_b$};

    \def \s{\x+2.4*\c}

    \draw [fill=gray!15] (\s, \c) rectangle (\s+\c, 0.1);
    \node at (\s+\c/2,\c/2) {$\sigma_{k'}$};
    
    \draw [fill=cyan!15] (\s, -\c) rectangle (\s+\c, -0.1);
    \node at (\s+\c/2,-\c/2) {$\sigma_{k'}$};

    \def \s{\x+3.6*\c}

    \draw [fill=gray!15] (\s, \c) rectangle (\s+\c, 0.1);
    \node at (\s+\c/2,\c/2) {$\sigma_a$};
    
    \draw [fill=cyan!15] (\s, -\c) rectangle (\s+\c, -0.1);
    \node at (\s+\c/2,-\c/2) {$\sigma_b$};

     \node at (\x+1.8*\c,-3*\c) {$ =  \sum_{k,k'}  u_{k}u_{k'}^* f(a,k)f(a,k')
    f(b,k)f(b,k')\, $};

    \def \x{13}

     \def \s{\x-1.2*\c}
     
    \def \a{\x-1.2*\c}
    \def \b{\x+1.2*\c}
    \draw (\a,\c/2-3*\c) -- (\b,\c/2-3*\c);
    \draw (\a,-\c/2-3*\c) -- (\b,-\c/2-3*\c);

    \draw [fill=gray!15] (\s, \c-3*\c) rectangle (\s+\c, 0.1-3*\c);
    \node at (\s+\c/2,\c/2-3*\c) {$\sigma_k$};
    
    \draw [fill=cyan!15] (\s, -\c-3*\c) rectangle (\s+\c, -0.1-3*\c);
    \node at (\s+\c/2,-\c/2-3*\c) {$\sigma_k$};

    \draw [fill=white!15] (\x , \c-3*\c) rectangle (\x+\c,-\c-3*\c);
    \node at (\x+\c/2,0-3*\c) {$\Tilde{\rho}$};

    \def \s{\x+1.2*\c}

    \draw [fill=gray!15] (\s, \c-3*\c) rectangle (\s+\c, 0.1-3*\c);
    \node at (\s+\c/2,\c/2-3*\c) {$\sigma_{k'}$};
    
    \draw [fill=cyan!15] (\s, -\c-3*\c) rectangle (\s+\c, -0.1-3*\c);
    \node at (\s+\c/2,-\c/2-3*\c) {$\sigma_{k'}$};

    \end{tikzpicture}
\end{equation}
with functions $f(i,j)$ that are either $+1$ or $-1$. 
The sign follows from the the anti-commutation relations between the Pauli matrices $\sigma_a, \sigma_k,\sigma_b$ and $\sigma_{k'}$.
Since, we are doing a quasiprobability simulation, we are not performing the whole sum over $k,k'$ and only apply a specific channel that depends on $k,k'$. 
We either have the case $k=k'$, but then the recovery is trivial or the case $k\neq k'$. 
In this case, we have knowledge of $k,k'$ and $a,b$ and can in the post processing weight the result with a factor of $f(a,k)f(a,k')f(b,k)f(b,k')$, which is either $+1$ or $-1$. Such a weighting does not affect the $\gamma$-factor, but gives the right estimator for the expectation value. Since all used channels only require local operations, we present a concrete construction with $\overline{\gamma}_{\LOCC}(U_{A_1,B_1}, V_{A_2,B_2}) \leq \overline{\gamma}_{\LO}(U_{A_1,B_1}, V_{A_2,B_2}) \leq 1 + 2 \sum_{\mathbf{k} \neq\mathbf{k}'} |w_{\mathbf{k}}| |w_{\mathbf{k}'}|$.

This construction works independently of the black box channel and has the same overhead as the parallel cut (the black box being the identity). It this therefore optimal according to our definition~\cref{def_blackbox_factor}.
We show the procedure as a quantum circuit in~\cref{fig_blackbox_protocol}.
\end{proof}

\begin{figure}
    \centering
        \begin{tikzpicture}[thick,scale=0.8]

    \def \x{0}
    
    \def \s{0.9}

     \draw [fill=gray!15,draw=none] (\x+0.05,0.06) rectangle (\x+4.0,1.5);   
     \draw [fill=cyan!15,draw=none] (\x+0.05,-0.06) rectangle (\x+4.0,-1.5);

     \draw [fill=gray!15,draw=none] (\s+0.1,1.4+0.4) rectangle (\s+0.1+0.25,1.4+0.25+0.4);
     \draw [fill=cyan!15,draw=none] (\s+1.6,1.4+0.4) rectangle (\s+1.6+0.25,1.4+0.25+0.4);     
     \node[gray] at (\s + 0.6,1.5+0.4) {\footnotesize{$A$}};
     \node[cyan] at (\s + 2.1,1.5+0.4) {\footnotesize{$B$}};

\node at (\x+2,0) {$----------$};

\node at (\x + 4.75 ,0) {\Large{'$=$'}};

     \draw (\x,0.5) -- (\x+4,0.5);
     \draw (\x,-0.5) -- (\x+4,-0.5); 
     
     \draw (\x,0.9) -- (\x+4,0.9);
     \draw (\x,-0.9) -- (\x+4,-0.9);

     \draw (\x,1.3) -- (\x+4,1.3);
     \draw (\x,-1.3) -- (\x+4,-1.3);

     \def \xx{0.8+\x}
     \draw [fill=red!40] (\xx,0.9) circle (2.5mm);
     \node at (\xx,0.9) {\footnotesize{$U_1$}};
     \draw (\xx,0.9-0.25) -- (\xx,-0.5+0.25);
     \draw [fill=red!40](\xx,-0.5) circle (2.5mm);
     \node at (\xx,-0.5) {\footnotesize{$U_1$}};

     \def \xx{2+\x}

     \draw [fill=blue!50] (\xx-0.3,1.4) -- (\xx+0.3,1.4) -- (\xx+0.5,0.4) -- (\xx+0.5,-0.4) -- (\xx+0.3,-1.4) -- (\xx-0.3,-1.4) -- (\xx-0.5,-0.4) -- (\xx-0.5,0.4) --cycle; 
     \node at (\xx,0) {\large{$\cE$}};

     \def \xx{3.1+\x}
     \draw [fill=Yellow!40] (\xx,1.3) circle (2.5mm);
     \node at (\xx,1.3) {\footnotesize{$U_2$}};
     \draw (\xx,1.3-0.25) -- (\xx,-1.3+0.25);
     \draw [fill=Yellow!40](\xx,-1.3) circle (2.5mm);
     \node at (\xx,-1.3) {\footnotesize{$U_2$}};

    \def\x{5.5}
    
    \def \s{1.5}

     \draw [fill=gray!15,draw=none] (\x+0.05,0.06) rectangle (\x+6.0,2.2);   
     \draw [fill=cyan!15,draw=none] (\x+0.05,-0.06) rectangle (\x+6.0,-2.2);

     \draw [fill=gray!15,draw=none] (\s+\x+0.1,2.1+0.4) rectangle (\s+\x+0.1+0.25,2.1+0.25+0.4);
     \draw [fill=cyan!15,draw=none] (\s+\x+1.6,2.1+0.4) rectangle (\s+\x+1.6+0.25,2.1+0.25+0.4);     
     \node[gray] at (\s+\x + 0.6,2.2+0.4) {\footnotesize{$A$}};
     \node[cyan] at (\s+\x + 2.1,2.2+0.4) {\footnotesize{$B$}};

\node at (\x+3,0) {$--------------$};



\def \xx{\x+0.3}

     \draw[decorate,decoration={snake,amplitude=.4mm,segment length=1mm,post length=1mm}] (\xx,1.7) -- (\xx,2.1);
     
     \draw[decorate,decoration={snake,amplitude=.4mm,segment length=1mm,post length=1mm}] (\xx,-1.7) -- (\xx,-2.1);
     
     \draw[fill=black] (\xx,1.7) circle (0.5mm);
     \draw[fill=black] (\xx,2.1) circle (0.5mm);

     \draw[fill=black] (\xx,-1.7) circle (0.5mm);
     \draw[fill=black] (\xx,-2.1) circle (0.5mm);

     \draw (\x,0.5) -- (\x+6,0.5);
     \draw (\x,-0.5) -- (\x+6,-0.5); 
     
     \draw (\x,0.9) -- (\x+6,0.9);
     \draw (\x,-0.9) -- (\x+6,-0.9);

     \draw (\x,1.3) -- (\x+4.3,1.3);
     \draw (\x,-1.3) -- (\x+4.3,-1.3);

     \draw (\x,1.7) -- (\x+4.3,1.7);
     \draw (\x,-1.7) -- (\x+4.3,-1.7);

     \draw (\x,2.1) -- (\x+6,2.1);
     \draw (\x,-2.1) -- (\x+6,-2.1);

     \def \g{5.1}

     \draw[thin] (\g+\x+0.16,2.1-0.2) -- (\g+\x+0.16,1.5);
     \draw[thin] (\g+\x+0.24,2.1-0.2) -- (\g+\x+0.24,1.42);

     \draw[thin] (\g+\x+0.24,1.42) -- (\g+\x-0.6,1.42);
     \draw[thin] (\g+\x+0.16,1.5) -- (\g+\x-0.6,1.5);

     \draw[thin] (\g+\x+0.16,-2.1+0.2) -- (\g+\x+0.16,-1.5);
     \draw[thin] (\g+\x+0.24,-2.1+0.2) -- (\g+\x+0.24,-1.42);

     \draw[thin] (\g+\x+0.16,-1.5) -- (\g+\x-0.6,-1.5);
     \draw[thin] (\g+\x+0.24,-1.42) -- (\g+\x-0.6,-1.42);
     
     \draw [fill=white!] (\g+\x,2.1-0.25) rectangle (\g+\x+0.5,2.1+0.25);
     \draw [fill=white!] (\g+\x,-2.1-0.25) rectangle (\g+\x+0.5,-2.1+0.25);
     \node at (\g+\x+0.25,2.1) {\scriptsize{$\sigma_a$}};
     \node at (\g+\x+0.25,-2.1) {\scriptsize{$\sigma_b$}};

     \def \g{4.3}
     \draw [fill=white!] (\g+\x,1.3-0.2) rectangle (\g+\x+0.4,1.7+0.2);
     \draw [fill=white!] (\g+\x,-1.7-0.2) rectangle (\g+\x+0.4,-1.3+0.2);

     \draw[thin,->] (\g+\x+0.2,1.4) -- (\g+\x+0.3,1.65);
     \draw[thin] (\g+\x+0.37,1.4) arc (0:180:0.17);
     \draw[thin,->] (\g+\x+0.2,-1.55) -- (\g+\x+0.3,-1.3);
     \draw[thin] (\g+\x+0.37,-1.55) arc (0:180:0.17);

     \def \xx{3.4+\x}

     \draw [fill=blue!50] (\xx-0.3,1.4) -- (\xx+0.3,1.4) -- (\xx+0.5,0.4) -- (\xx+0.5,-0.4) -- (\xx+0.3,-1.4) -- (\xx-0.3,-1.4) -- (\xx-0.5,-0.4) -- (\xx-0.5,0.4) --cycle; 
     \node at (\xx,0) {\large{$\cE$}};

     \def \xx{1+\x}
     \draw [fill=red!40] (\xx,0.9) circle (2.5mm);
     \node at (\xx,0.9) {\footnotesize{$U_1$}};
     \draw (\xx,0.9-0.25) -- (\xx,-0.5+0.25);
     \draw [fill=red!40](\xx,-0.5) circle (2.5mm);
     \node at (\xx,-0.5) {\footnotesize{$U_1$}};

     \def \xx{2+\x}
     \draw [fill=Yellow!40] (\xx,2.1) circle (2.5mm);
     \node at (\xx,2.1) {\footnotesize{$U_2$}};
     \draw (\xx,2.1-0.25) -- (\xx,-2.1+0.25);
     \draw [fill=Yellow!40](\xx,-2.1) circle (2.5mm);
     \node at (\xx,-2.1) {\footnotesize{$U_2$}};

    \end{tikzpicture}
    \caption{Protocol of the black box cut. Both parties $A$ and $B$ apply a Pauli operation $\sigma_x$ based on their Bell-measurement outcome.}
    \label{fig_blackbox_protocol}
\end{figure}

In the scenario of normal gate teleportation, it is critical that we have to send classical information to the other party. Otherwise we would be able to signal faster than light. In the QPD setting, this is not an issue, since each individual channel cannot signal. 

\subsection{Cutting in the black box setting for N two-qubit gates}\label{sec_black_box_n}
For completeness, we will present how to generalize the previous theorem to $n$ gates. For this, we first define the black box overhead for multiple channels.
\begin{definition}[Black box $\gamma$-factor] \label{def_blackbox_factor_n}
    The minimal black box overhead of channels $\{\cE_{AB}^{(j)}\}_j$ that are separated by arbitrary black box channels $\{X_{AB}^{(j)}\}_j$ is given by
  \begin{align}
    \overline{\gamma}_S(\{\cE_{AB}^j\}_j ) 
      &\coloneqq \min \Big\{ \sum_{i=1}^m \lvert a_i \rvert :  \prod_j   X_{AB}^{(j)} \cE_{AB}^{(j)} =  \tr_{E_AE_B} \sum\limits_{i=1}^m a_i (\prod_j  X_{AB}^{(j)} \cF^{(j)}_i) , \nonumber \\ 
      &\hspace{25mm} \forall X_{AB}^{(j)}\in\CPTP, \, m\geq 1, \, \cF_i^{(j)} \in S(AE_A,BE_B) \textnormal{ and } a_i \in \R \Big\} \, ,
  \end{align} 
  with $S(AE_A,BE_B)=\{ \LO(AE_A,BE_B),\LOCC(AE_A,BE_B) \}$.
\end{definition}

Now we extend the previous theorem.
\begin{theorem}\label{theorem:blackbox_n}
    $N$ two-qubit unitaries $\{U_{A_i,B_i}^{(i)}\}_{i=1}^N$ have a decomposition independent of interleaving black box channel $X_{AB}^{(i)}$, with optimal black box overhead
    \begin{equation}
        \overline{\gamma}_\LOCC(\{U_{A_i,B_i}^{(i)}\}_i^N) = \overline{\gamma}_\LO(\{U_{A_i,B_i}^{(i)}\}_i^N) = 1 + 2 \sum_{\mathbf{k} \neq\mathbf{k}'} |w_{\mathbf{k}}| |w_{\mathbf{k}'}| \, ,
    \end{equation}
    where $\mathbf{k} \in \{0,1,2,3\}^N$ and $w_{\mathbf{k}} = \prod_{i=1}^N  u^{(i)}_{k_i}$ is the product of the KAK-coefficients of $U_{A_i,B_i}^{(i)}$.
\end{theorem}

Note that~\cref{theorem:blackbox} and~\cref{theorem:blackbox_n} have the same simulation overhead as~\cref{cor_parallel_cuts}. Arbitrary cuts are therefore, apart from a requirement for ancillary qubits, not more expensive than parallel cuts.

\begin{proof}
    We proof the assertion by explicitly constructing operations $\cF_i^{(j)} \in \LO(AE_A,BE_B)$. Optimality then follows from the black boxes being identity maps and \cref{cor_parallel_cuts}.
    The procedure is analogous to the two gate case. First we implement the channel used for the joint parallel cut of all $n$ two-qubit gates (cf. \cref{sec_parallel_2qubit}) via $\cF_i^{(j)}$. 
    This comes at a sampling overhead of
    \begin{equation}
         1 + 2 \sum_{\mathbf{k} \neq\mathbf{k}'} |w_{\mathbf{k}}| |w_{\mathbf{k}'}| \, ,
    \end{equation}
    where $\mathbf{k} \in \{0,1,2,3\}^N$ and $w_{\mathbf{k}} = \prod_{i=1}^N  u^{(i)}_{k_i}$ is the product of the KAK-coefficients of $U_{A_i,B_i}^{(i)}$.
    Importantly, instead of implementing the two-qubit unitaries \smash{$U_{A_i,B_i}^{(i)}$} at their position $A_i$, $B_i$ we apply them to Bell states in the environment $E_A$ and $E_B$ as seen in the previous part of this section. These ``virtual'' Choi states can then be used to apply the two-qubit gate at a later time. This is done via the operations $\{ \cF_i^{(j)}\}_{i=1}^{j-1}$. As seen before, the strategy consists of a bell measurement between $E_A$ and $A_i$ and $E_B$ and $B_i$ together with a corresponding local correction. The overall sign of the sampling run can be calculated based on the measurement results. Since all this can be done locally and in the postprocessing, these operations do not change the overhead. 
\end{proof}


\section{Open questions} \label{sec_open_questions}
While our work now provides an extended understanding of circuit cutting for two-qubit unitaries, it still remains an open question to characterize the optimal overhead for more general unitaries.
For some specific unitaries like the Toffoli gate, our results can be naturally extended, as seen in~\cref{rmk_toffoli} below.
\begin{remark}[Toffoli gate] \label{rmk_toffoli}
We can apply~\cref{cor_two_qubit_gamma} to obtain optimal cuts for the Toffoli gate via the two identities given in~\cref{fig_toffoli}.
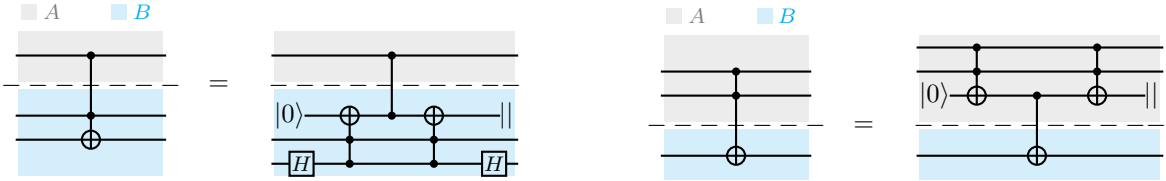
\begin{figure}[!htb]
    \centering
    \begin{subfigure}[b]{0.45\textwidth}
            \begin{tikzpicture}[thick,scale=0.8]
    \def\x{4.25}

     \draw [fill=gray!15,draw=none] (0.05,0.06) rectangle (2.45,0.9);   
     \draw [fill=cyan!15,draw=none] (0.05,-0.06) rectangle (2.45,-1.5);  

     \draw [fill=gray!15,draw=none] (\x+0.05,0.06) rectangle (\x+4.05,0.9);   
     \draw [fill=cyan!15,draw=none] (\x+0.05,-0.06) rectangle (\x+4.05,-1.5);

     \draw [fill=gray!15,draw=none] (0.1,1+0.1) rectangle (0.1+0.25,1+0.25+0.1);
     \draw [fill=cyan!15,draw=none] (1.6,1+0.1) rectangle (1.6+0.25,1+0.25+0.1);     
     \node[gray] at (0.6,1.125+0.1) {\footnotesize{$A$}};
     \node[cyan] at (2.1,1.125+0.1) {\footnotesize{$B$}};
    
     \draw (0,0.5) -- (2.5,0.5);
     \draw (0,0.-0.9) -- (2.5,-0.9);
     \draw (0,-0.5) -- (2.5,-0.5);   

\node at (1.25,0) {$-------$};

     \def\xx{1.25}
     \draw[fill=black] (\xx,0.5) circle (0.5mm);
     \draw[fill=black] (\xx,-0.5) circle (0.5mm);
     \draw (\xx,-0.9) circle (1.5mm);
     \draw (\xx,0.5) -- (\xx,-0.9-0.15);

     \node at ((1.25+\x/2,0) {$=$};
    \def\x{4.25}
    \draw (\x,-0.9) -- (\x+4.1,-0.9);
    \draw (\x,0.5) -- (\x+4.1,0.5);
    \draw (\x+0.55,-0.5) -- (\x+3.8,-0.5);
    \draw (\x,-1.3) -- (\x+4.1,-1.3);

    \node at (\x+0.3,-0.5) {$\ket{0}$};
    \node at (\x+3.9,-0.5) {$||$};

    \draw[fill=cyan!15] (\x+0.3,-1.3-0.2) rectangle (\x+0.3+0.4,-1.3+0.2);
    \node at (\x+0.3+0.2,-1.3) {\footnotesize{$H$}};

    \draw[fill=cyan!15] (\x+3.5,-1.3-0.2) rectangle (\x+3.5+0.4,-1.3+0.2);
    \node at (\x+3.5+0.2,-1.3) {\footnotesize{$H$}};
\node at (\x+2.05,0) {$----------$};

    \def\xx{\x+1.3}
     \draw[fill=black] (\xx,-1.3) circle (0.5mm);
     \draw[fill=black] (\xx,-0.9) circle (0.5mm);
     \draw (\xx,-0.5) circle (1.5mm);
     \draw (\xx,-1.3) -- (\xx,-0.5+0.15);

    \def\xx{\x+2}
     \draw[fill=black] (\xx,0.5) circle (0.5mm);
     \draw[fill=black] (\xx,-0.5) circle (0.5mm);
     \draw (\xx,0.5) -- (\xx,-0.5);

    \def\xx{\x+2.7}
     \draw[fill=black] (\xx,-1.3) circle (0.5mm);
     \draw[fill=black] (\xx,-0.9) circle (0.5mm);
     \draw (\xx,-0.5) circle (1.5mm);
     \draw (\xx,-1.3) -- (\xx,-0.5+0.15);

    \end{tikzpicture}
        \caption{Cutting a Toffoli between qubits 1 and 2.}
        \label{fig:toff12}
    \end{subfigure}
        \hfill
    \begin{subfigure}[b]{0.45\textwidth}
            \begin{tikzpicture}[thick,scale=0.8]
    \def\x{4.25}

     \draw [fill=gray!15,draw=none] (0.05,0.06) rectangle (2.45,1.5);   
     \draw [fill=cyan!15,draw=none] (0.05,-0.06) rectangle (2.45,-0.9);  

     \draw [fill=gray!15,draw=none] (\x+0.05,0.06) rectangle (\x+4.05,1.5);   
     \draw [fill=cyan!15,draw=none] (\x+0.05,-0.06) rectangle (\x+4.05,-0.9);

     \draw [fill=gray!15,draw=none] (0.1,1+0.7) rectangle (0.1+0.25,1+0.25+0.7);
     \draw [fill=cyan!15,draw=none] (1.6,1+0.7) rectangle (1.6+0.25,1+0.25+0.7);     
     \node[gray] at (0.6,1.125+0.7) {\footnotesize{$A$}};
     \node[cyan] at (2.1,1.125+0.7) {\footnotesize{$B$}};
    
     \draw (0,0.9) -- (2.5,0.9);
     \draw (0,0.5) -- (2.5,0.5);
     \draw (0,-0.5) -- (2.5,-0.5);   

\node at (1.25,0) {$-------$};

     \def\xx{1.25}
     \draw[fill=black] (\xx,0.9) circle (0.5mm);
     \draw[fill=black] (\xx,0.5) circle (0.5mm);
     \draw (\xx,-0.5) circle (1.5mm);
     \draw (\xx,0.9) -- (\xx,-0.5-0.15);
     
     \node at ((1.25+\x/2,0) {$=$};
    \def\x{4.25}
    \draw (\x,1.3) -- (\x+4.1,1.3);
    \draw (\x,0.9) -- (\x+4.1,0.9);
    \draw (\x+0.55,0.5) -- (\x+3.8,0.5);
    \draw (\x,-0.5) -- (\x+4.1,-0.5);

    \node at (\x+0.3,0.5) {$\ket{0}$};
    \node at (\x+3.9,0.5) {$||$};

\node at (\x+2.05,0) {$----------$};

    \def\xx{\x+1}
     \draw[fill=black] (\xx,1.3) circle (0.5mm);
     \draw[fill=black] (\xx,0.9) circle (0.5mm);
     \draw (\xx,0.5) circle (1.5mm);
     \draw (\xx,1.3) -- (\xx,0.5-0.15);

    \def\xx{\x+2}
     \draw[fill=black] (\xx,0.5) circle (0.5mm);
     \draw (\xx,-0.5) circle (1.5mm);
     \draw (\xx,0.5) -- (\xx,-0.5-0.15);

     \def\xx{\x+3}
     \draw[fill=black] (\xx,1.3) circle (0.5mm);
     \draw[fill=black] (\xx,0.9) circle (0.5mm);
     \draw (\xx,0.5) circle (1.5mm);
     \draw (\xx,1.3) -- (\xx,0.5-0.15);

    \end{tikzpicture}
        \caption{Cutting a Toffoli between qubits 2 and 3.}
        \label{fig:toff23}
    \end{subfigure}        
\caption{These two circuits allow us to prove that $\gamma(\mathrm{Toffoli})=3$ as described in the text below.}
\label{fig_toffoli}
\end{figure}
We immediately see that $\gamma(\mathrm{Toffoli}) \leq 3$, since $\gamma(\mathrm{CZ})=\gamma(\mathrm{CNOT})=3$ as ensured by~\cref{cor_two_qubit_gamma}.
The other direction, i.e.~$\gamma(\mathrm{Toffoli}) \geq 3$ follows since $\gamma(\mathrm{Toffoli}) \geq \gamma(\proj{\psi^{(\phi)}_{\mathrm{Toffoli}}})$, where 
\begin{align}
\ket{\psi^{(\phi)}_{\mathrm{Toffoli}}}=( \mathrm{Toffoli}_{AAB} \ox \mathds{1}_{B} )(\ket{+}_A \ket{1}_{A} \otimes \ket{0}_B\ket{0}_{B}) \, .
\end{align}
A small calculation then shows
\begin{align}
\gamma(\mathrm{Toffoli}) \geq\gamma(\proj{\psi^{(\phi)}_{\mathrm{Toffoli}}}) = 3 \, .
\end{align}
Surprisingly, $\gamma(\ket{\mathrm{Choi_{Toffoli}}})=2.76$, indicating, that the Choi state is in general not a tight bound from below.
Similar identities can be found for larger multi-controlled gates based on the symmetry of the CCZ gate.
\end{remark}
If similar arguments can be made for arbitrary unitaries is unclear.
An interesting step in this direction could be to isolate interaction parts.
For instance, for the Toffoli gate, it is possible to isolate the non-local interaction, even though it is not of KAK-like form and~\cref{Master} does not apply directly. It would be interesting to see for which other unitaries this is the case.

Another important topic for further research is the characterization of the optimal sampling overhead for non-unitary channels.
Here it seems likely that there is a strict separation between $\gamma_{\LO}$ and $\gamma_{\LOCC}$, at least on channels that themselves involve classical communication between the bipartitions.

\paragraph{Relation to other works}
During the preparation of this work, we became aware of two other independent efforts to optimally cut two-qubit unitaries which were published at almost the same time~\cite{ufrecht2023optimal,anguspaper}.
The authors in~\cite{ufrecht2023optimal} obtain the same optimality result for the single-cut, parallel cut, and black box setting, however they only consider controlled rotation gates and not general two-qubit gates.
In~\cite{anguspaper}, the authors derive the same technical result as~\cref{Master} and thus obtain the same optimality result for single cuts and parallel cuts of arbitrary two-qubit gates. However, they do not investigate the optimal overhead for the black box setting.

 \paragraph{Acknowledgements}
 We thank Stefan Woerner for help with~\cref{rmk_toffoli} and Angus Lowe for useful discussions.
This work was supported by the ETH Quantum Center, the National Centres for Competence in Research in Quantum Science and Technology (QSIT) and The Mathematics of Physics (SwissMAP), and the Swiss National Science Foundation Sinergia grant CRSII5 186364.
 

\appendix
\section{Circuits for implementing cutting protocol}\label{sec:circ}
To obtain the optimal sampling overhead presented in \cref{Master}, one has to apply 
the channel $({\cC_{k,k'}^{\phi}})_\pm(\rho) \otimes {\cD_{k,k'}^{\phi}})_\pm(\rho)$. 
This channel is implemented by the following circuit:
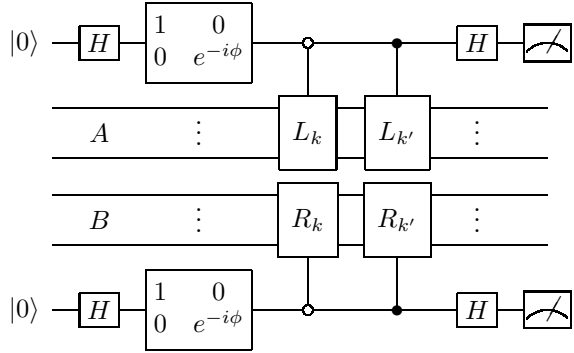
\begin{equation*}
			\Qcircuit @C=1em @R=0.5em{
				\lstick{\ket{0}} & \gate{H} & \gate{\begin{matrix} 1 & 0 \\ 0 & e^{-i\phi} \end{matrix}} & \ctrlo{1} & \ctrl{1} &  \gate{H}  & \meter\\
				 & \qw & \qw & \multigate{2}{L_k}  & \multigate{2}{L_{k'}}  \qw & \qw & \qw \\
				 & A & \raisebox{0.4\height}{\vdots} &   &  & \raisebox{0.4\height}{\vdots} & \\ 
               & \qw  & \qw &  \ghost{L_k} &  \ghost{L_{k'}}  & \qw  & \qw \\
				 & \qw  & \qw &  \ghost{R_k} &  \ghost{R_{k'}}  & \qw  &  \qw \\
              & B & \raisebox{0.4\height}{\vdots} &   &  & \raisebox{0.4\height}{\vdots} & \\ 
              & \qw & \qw & \multigate{-2}{R_k} & \multigate{-2}{R_{k'}} &   \qw & \qw \\
              \lstick{\ket{0}} & \gate{H}  &  \gate{\begin{matrix} 1 & 0 \\ 0 & e^{-i\phi} \end{matrix}} & \ctrlo{-1} & \ctrl{-1}  &    \gate{H}  &    \meter 
			}
\end{equation*}
The first and the last qubit are ancillary qubits that are required for the negativity trick. Their outcome $a_1$ and $a_2$ ($a_i \in \{0,1\}$) determines the sign $(-1)^{a_1+a_2}$ with which the shot is weighted.
With this the total workflow can be described by the following two steps:
\begin{enumerate}[(1)]
    \item Determine the KAK-like decomposition of the unitary $U_{AB}$ that is to be cut. \begin{equation*}
        U_{AB} = \left(V_A^{(1)}\otimes V_B^{(2)}\right)  \sum_{k=1}^N u_{k} (L_k)_A \otimes (R_k)_B  \left(V_A^{(3)}\otimes V_B^{(4)}\right)
    \end{equation*}
    In the case of parallel two-qubit unitaries, this can be easily achieved by multiplying their single KAK-decompositions.
    \item Randomly run the following circuits and weight the measurement outcome with the corresponding factor in post processing:
\end{enumerate}
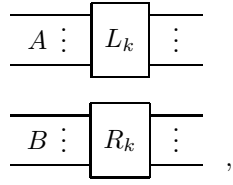
\begin{itemize}
    \item \begin{equation*}
			\Qcircuit @C=1em @R=0.5em{
				 & \qw & \qw & \multigate{2}{L_k} &  \qw & \qw \\
				 & A & \raisebox{0.4\height}{\vdots} &   &  \raisebox{0.4\height}{\vdots} & \\ 
               & \qw  & \qw &  \ghost{L_k}  & \qw  & \qw \\
                &   &  &    &  &  \\
				 & \qw  & \qw &  \ghost{R_k}  & \qw  &  \qw \\
              & B & \raisebox{0.4\height}{\vdots} &   &  \raisebox{0.4\height}{\vdots} & \\ 
              & \qw & \qw & \multigate{-2}{R_k} &  \qw & \qw &,
			} 
\end{equation*}
The measurement outcome of each of these $N$ circuits is weighted by the KAK-coefficient $|u_k|^2$.
\item There are $N(N-1)/2$ circuits of the form
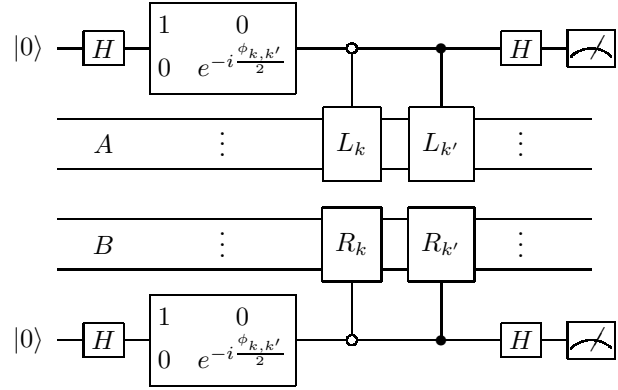
\begin{equation*}
			\Qcircuit @C=1em @R=0.5em{
				\lstick{\ket{0}} & \gate{H} & \gate{\begin{matrix} 1 & 0 \\ 0 & e^{-i\frac{\phi_{k,{k'}}}{2}} \end{matrix}}
 & \ctrlo{1} & \ctrl{1} &  \gate{H}  & \meter\\
				 & \qw & \qw & \multigate{2}{L_k}  & \multigate{2}{L_{k'}}  \qw & \qw & \qw \\
				 & A & \raisebox{0.4\height}{\vdots} &   &  & \raisebox{0.4\height}{\vdots} & \\ 
               & \qw  & \qw &  \ghost{L_k} &  \ghost{L_{k'}}  & \qw  & \qw \\
				 &   &  &    &  &  \\ 
            & \qw  & \qw &  \ghost{R_k} &  \ghost{R_{k'}}  & \qw  &  \qw \\
              & B & \raisebox{0.4\height}{\vdots} &   &  & \raisebox{0.4\height}{\vdots} & \\ 
              & \qw & \qw & \multigate{-2}{R_k} & \multigate{-2}{R_{k'}} &   \qw & \qw \\
              \lstick{\ket{0}} & \gate{H}  &  \gate{\begin{matrix} 1 & 0 \\ 0 & e^{-i\frac{\phi_{k,{k'}}}{2}} \end{matrix}}
 & \ctrlo{-1} & \ctrl{-1}  &    \gate{H}  &    \meter 
			}
\end{equation*}
with $k<k'$,
that have to be weighted by $2 |u_k||u_{k'}| (-1)^{a_1+a_2}$.
\item There $N(N-1)/2$ circuits of the form
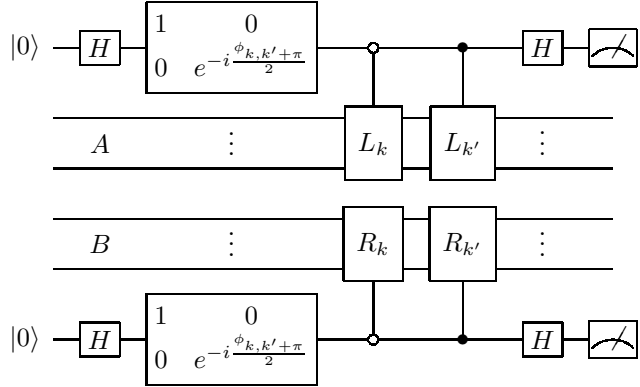
\begin{equation*}
			\Qcircuit @C=1em @R=0.5em{
				\lstick{\ket{0}} & \gate{H} & \gate{\begin{matrix} 1 & 0 \\ 0 & e^{-i\frac{\phi_{k,{k'}+\pi}}{2}} \end{matrix}} & \ctrlo{1} & \ctrl{1} &  \gate{H}  & \meter\\
				 & \qw & \qw & \multigate{2}{L_k}  & \multigate{2}{L_{k'}}  \qw & \qw & \qw \\
				 & A & \raisebox{0.4\height}{\vdots} &   &  & \raisebox{0.4\height}{\vdots} & \\ 
               & \qw  & \qw &  \ghost{L_k} &  \ghost{L_{k'}}  & \qw  & \qw \\
		      &   &  &    &  &  \\
            & \qw  & \qw &  \ghost{R_k} &  \ghost{R_{k'}}  & \qw  &  \qw \\
              & B & \raisebox{0.4\height}{\vdots} &   &  & \raisebox{0.4\height}{\vdots} & \\ 
              & \qw & \qw & \multigate{-2}{R_k} & \multigate{-2}{R_{k'}} &   \qw & \qw \\
              \lstick{\ket{0}} & \gate{H}  &  \gate{\begin{matrix} 1 & 0 \\ 0 & e^{-i\frac{\phi_{k,{k'}+\pi}}{2}} \end{matrix}} & \ctrlo{-1} & \ctrl{-1}  &    \gate{H}  &    \meter 
			}
\end{equation*}
with $k<k'$,
that have to be weighted by $-2 |u_k||u_{k'}| (-1)^{a_1+a_2}$ . 
\end{itemize}
The  phases $\phi_{k,k'}$ are determined via the phases of the coefficients $u_k = |u_k|\ee^{\ci\phi_k}$ and $\phi_{k,k'} = \phi_k - \phi_k'$.


\bibliographystyle{arxiv_no_month}
\bibliography{bibliofile}

\end{document}